\documentclass[pdflatex,sn-mathphys-ay]{sn-jnl}

\usepackage{graphicx}%
\usepackage{multirow}%
\usepackage{amsmath,amssymb,amsfonts}%
\usepackage{amsthm}%
\usepackage{mathrsfs}%
\usepackage[title]{appendix}%
\usepackage{xcolor}%
\usepackage{textcomp}%
\usepackage{manyfoot}%
\usepackage{booktabs}%
\usepackage{algorithm}%
\usepackage{algorithmicx}%
\usepackage{algpseudocode}%
\usepackage{listings}%

\usepackage{lineno}

\usepackage{float}
\usepackage{caption}
\usepackage{subcaption}
\usepackage{tikz}
\usetikzlibrary{arrows.meta,positioning,calc}
\usepackage{pgfplots}
\pgfplotsset{compat=1.18}
\usepgfplotslibrary{groupplots}
\usepackage{comment}

\theoremstyle{thmstyleone}%
\newtheorem{theorem}{Theorem}[section]

\newtheorem{lemma}{Lemma}[section]

\newtheorem{conjecture}{Conjecture}[section]
\newtheorem{remark}{Remark}[section]
\theoremstyle{remark}

\newcommand{\R}{\mathbb{R}}
\newcommand{\one}{\mathbf{1}}
\newcommand{\diag}{\mathrm{diag}}
\begin{document}
%\begin{linenumbers}

\title[Article Title]{Bounds on $R_0$ and final epidemic size when the next-generation matrix $M$ is only partially known}

\author[1,2]{\fnm{Andrea} \sur{Bizzotto}}\email{andreabizzotto.w@gmail.com}
\equalcont{These authors contributed equally to this work.}
\author[3]{\fnm{Frank} \sur{Ball}}\email{frank.ball@nottingham.ac.uk}
\equalcont{These authors contributed equally to this work.}
\author*[4]{\fnm{Tom} \sur{Britton}}\email{tom.britton@math.su.se}
\equalcont{These authors contributed equally to this work.}

\affil[1]{\orgdiv{Health Emergencies Center}, \orgname{Bruno Kessler Foundation}, \city{Trento}, \country{Italy}}
\affil[2]{\orgdiv{Department of Mathematics}, \orgname{University of Trento}, \city{Trento}, \country{Italy}}
\affil[3]{\orgdiv{School of Mathematical Sciences}, \orgname{University of Nottingham}, \city{Nottingham}, \country{United Kingdom}}
\affil[4]{\orgdiv{Department of Mathematics}, \orgname{Stockholm University}, \city{Stockholm}, \country{Sweden}}

\abstract{\unboldmath We study a multitype SIR epidemic model where individuals are categorized into different types, and where infection spread is characterized by a next-generation matrix $M = \{m_{ij}\}$ with community fractions $\{\pi_j\}$ for the different types of individuals. We analyse two key quantities: the basic reproduction number $R_0$ and the 
final epidemic outcome of the different types $\{ \tau_i\}$. We consider the situation where $M$ is only partly known, through the row sums $\{r_i\}$ or the column sums $\{c_j\}$, and treat both a general $M$ and the special but common situation where $M$ is proportional to a contact matrix satisfying detailed balance. For a general $M$, which is partially observed through $\{r_i\}$ or $\{c_j\}$, we obtain sharp upper and lower bounds of $R_0$ and $\{ \tau_i\}$, but for the case where $M$ satisfies detailed balance the problem is harder:
our obtained bounds for $R_0$ are narrower than the general case but still not sharp, and bounds for the final size are only obtained when there are two types of individual. }


\maketitle

\section{Introduction}\label{sec1}
A common situation in epidemic modelling is to consider different \emph{types} of individuals, where types may differ in infectivity, susceptibility and/or in how they mix with other types (e.g.\ \cite{DHB13}, Ch 7). A central quantity in such cases is the next-generation matrix $M$, with elements $m_{ij}$ denoting the average (expected) number of infectious contacts that an $i$-individual has with $j$-individuals during its infectious period (note that some authors reverse the indices). In Section \ref{Sec-Prel}, we present known results for multitype epidemics, including the basic reproduction number $R_0$ and the vector of final epidemic sizes $\{\tau_i\}$, and how they depend on $M$.

The outcome of such multitype epidemic models hence requires knowledge about $M$. 
One approach to learn about $M$, started already by \ \cite{M08}, has been Social Contact Studies (SCS). In SCS, sampled individuals are followed during e.g.\ 24 hours and all close contacts with other individuals are registered, containing also covariate information about the contactee including their age. 
In \cite{SCD} a large number of SCS are stored and available for further analysis. The most common type-categorization is age groups, 
and the resulting contact matrix from a SCS then denotes the average number of contacts between different age groups during one day. An infectious disease having transmission probability $p$ (between all types) and mean infectious period $\mu$ (for all types) would then have next-generation matrix $M=p\mu C$, where $C$ is the inferred contact matrix. 

There are however situations where the next-generation matrix is only partially known, and then a range of possible epidemic outcomes may be consistent with the partial knowledge of $M$. One recent such application is in  \cite{BB25}, where they show empirically that individuals differ in contact patterns also in other aspects than age, for example in having different social activity levels, thus distinguishing types both according to age but also other covariate information. In this situation the specific type of sampled individuals is often known, but typically not for the contactees. Consequently, $M$ is only partially known  through the row sums of $M$, corresponding to the overall mean number of contacts that individuals of each type have. 

The situation where $M$ is only partially known is the focus of the present paper: specifically, \emph{how much can we infer about epidemic outcomes when only limited information on contact patterns is available?} We focus on the situations where we either know the row sums $\{ r_i\}$, where $r_i=\sum_jm_{ij}$, or the column sums $\{c_j\}$, where $c_j=\sum_im_{ij}$ of $M$.

In our analyses we derive bounds on the basic reproduction number $R_0$, the final size in each type $\{\tau_i\}$, and the total final size $\bar{\tau}=\sum_{i=1}^{k}\pi_i\tau_i$, where $\pi_i$ is the community fraction of type-$i$ individuals.  These bounds are valid for all $M$ having some fixed given row sums $\{ r_i\}$ or fixed given column sums $\{ c_j\}$. This contrasts with studies such as \cite{CLANCY} and \cite{VIGGO}, which explore the effect of heterogeneity on $\bar{\tau}$ for fixed $R_0>1$. 
A related real-data application of next-generation matrix row and column sums was presented in \cite{Ang}, where column sums were used to quantify the transmission potential of specific age groups, and row sums their overall susceptibility to infection.

The rest of the paper is structured as follows.
In Section \ref{Sec-Prel}, we present known results for multitype epidemics. Our main results are given in Section \ref{Sec-Main}, with the proofs deferred to Section \ref{Sec-Proof}. In Section \ref{Sec-Illustr}, we illustrate our results with some theoretical as well as empirical examples.

\section{Preliminaries and notation}\label{Sec-Prel}
We now present well-known results for deterministic homogeneous and multitype SIR epidemic models. More details can be found in e.g.\ \cite{DHB13} (Section 1.3 and 13.2, respectively) or \cite{MA}. We start with the homogeneous case.
Throughout, bold symbols denote vectors, e.g.\ $\boldsymbol{\tau}=(\tau_1,\dots,\tau_k)^{\top}$ or $\boldsymbol{\pi}=(\pi_1,\dots,\pi_k)^{\top}$, where $\top$ denotes transpose. We write $\boldsymbol{1}$ (or $\boldsymbol{0}$) for the vector in $\mathbb{R}^k$ with all entries equal to $1$ (or $0$). For any vector $\boldsymbol{x}=(x_1,\dots,x_k)^{\top}$, we denote by $D_{\boldsymbol{x}}$ the diagonal matrix with diagonal entries $x_1,\dots,x_k$.

\subsection{$R_0$ and final size for homogeneous case}
An SIR epidemic model assumes that individuals are at first Susceptible to the disease. If they get infected they immediately become Infectious (and have infectious contacts with others), and after some time they recover and become immune for the rest of the study period. Let $m$ denote the average or expected number of infectious contacts made by an infectious individual. (For SIR models $m$ is often expressed as $p\beta \mu$, where $p$ is the transmission probability, $\beta$ is the contact rate and $\mu$ is the average duration of the infectious period). Each infectious contact is with an individual chosen independently and uniformly at random from the community, and only results in infection if the contacted person is still susceptible. 
For this model, the quantity $m$ is called the basic reproduction number and denoted $R_0=m$, and the final fraction infected $\tau$ solves the equation
\begin{equation}
\label{eq:finalsize-homogeneous}
    1-\tau =\rm{e}^{-R_0\tau}.
\end{equation}
The solution $\tau=0$ always exists (corresponding to a minor outbreak), but when $R_0>1$ there is a unique positive solution $\tau\in (0,1)$ determining the final size in the event of major outbreak. For later use, we define 
\begin{equation}   
  t_{\alpha} \text{ as the largest solution to the equation } 1-t=\rm{e}^{-\alpha t}.  \label{t-alpha}
\end{equation}
Note that $t_{\alpha}=0$ for $0\le \alpha \le 1$, whereas $0<t_{\alpha}<1$ for $\alpha >1$, and $t_{\alpha}$ is strictly increasing in $\alpha$ on $(1,\infty )$.

\subsection{$R_0$ and final size for multitype epidemics}\label{subsec2.2}
In the multitype setting we have $k$ different \emph{types} of individuals, and still consider an SIR class of epidemic (the current results in fact hold under more general situations, including SEIR and time varying infectivity settings). The mean number of infectious contacts now depends on which type of individuals are considered. The quantity $m$ hence converts to $M=(m_{ij})$, where $m_{ij}\ge 0$ denotes the average number of infectious contacts with $j$-individuals that one infectious $i$-individual has during its infectious period (again some authors reverse the indices). This quantity depends both on the propensity with which an infectious $i$-individual has contact with separate $j$-individuals, but also on the abundance of $j$-individuals. Hence, an alternative way of writing $m_{ij}$ is $m_{ij}=b_{ij}\pi_j$, where $b_{ij}/n$ ($n$ being the total population size) reflects the per-individual contact probability, and $\pi_j$ denotes the community fraction of $j$-individuals.  We assume that $\pi_j>0$ for all $j$.

The matrix $M$ is often referred to as the next-generation matrix (NGM). The expected number of infections after $r$ generations in a large initially susceptible population is proportional to $M^r$. This quantity grows with $r$ if and only if the largest eigenvalue of $M$ exceeds the value 1. Consequently, the basic reproduction number $R_0$ of the multitype epidemic is defined as the largest eigenvalue of $M$. We assume that $M$ is irreducible, meaning that each type can infect all other types (directly or indirectly through other types) -- otherwise we cannot speak of a common community. One situation considered below, when each type only mixes with itself, does not obey this assumption. The results should then be interpreted as a limit when close to all contacts are with the same type, but a tiny fraction of contacts are with all other types.
The final size vector $\boldsymbol{\tau}=(\tau_1,\dots , \tau_k)^{\top}$ solves the $k$ final size equations:
\begin{align}
\label{eq:finalsize}
1-\tau_j &= {\rm e}^{-\sum_{i=1}^k\pi_i\tau_im_{ij}/\pi_j}  ,\quad j=1,\dots ,k.
\end{align}
Equivalently, in vector-matrix form:
\begin{equation}
\label{eq:finalsize-vector}
    \boldsymbol{1}-\boldsymbol{\tau} = \rm{e}^{-D_{\boldsymbol{\pi}}^{-1}M^{\top}D_{\boldsymbol{\pi}}\;\boldsymbol{\tau}}.
\end{equation}
If $R_0\le 1$, then $(\tau_1,\dots , \tau_k)=(0,\dots , 0)$ is the only solution (minor outbreak) but when $R_0>1$ there is a unique positive solution $\boldsymbol{\tau}$ in the event of major outbreak, more details can be found in e.g. \cite{S-T}, \cite{VIGGO} or \cite{MSW18}. Note that the expression $m_{ij}/\pi_j$ in the exponent of the  $j^{\rm th}$ equation may be replaced by $b_{ij}$ as defined above.

\subsection{Special forms of $M$ and notation}\label{Sec-spec-form}
We now present some special forms for $M$. We consider a $k\times k$ next-generation matrix $M=\{m_{ij}\}$ which may also be written as $B D_{\boldsymbol{\pi}}=\{ b_{ij}\pi_j\}$, so $b_{ij}=m_{ij}/ \pi_j$. 
%where $\pi_j$ denotes the community fraction of $j$-individuals. 
The row sums $\{ r_i\}$  and the column sums $\{ c_j\}$ are defined by $r_i=(M\boldsymbol{1})_i=\sum_j m_{ij}$ and $c_j=(M^\top\boldsymbol{1})_j=\sum_i m_{ij}$, and we define $r_{\max} =\max_i r_i$, $r_{\min} =\min_i r_i$, $c_{\max} =\max_j c_j$, and $c_{\min} =\min_jc_j$.
The next-generation matrix $M$ is most often estimated from a social contact survey, which captures how many contacts different types have, and with whom. On top of this contact heterogeneity one may either assume that all individuals are equally susceptible and equally infectious given a contact, or else that individuals differ also in this latter aspect. In the former, more common, situation where only contact heterogeneity is considered, $M$ satisfies \emph{detailed balance}: $\pi_im_{ij} = \pi_jm_{ji}$ for all $i,j$. This follows because contacts are symmetric, so the total number of contacts from $i$-individuals to $j$-individuals must equal the total number of contacts from $j$-individuals to $i$-individuals. If also susceptibilities and infectivities may vary between individuals, detailed balance need not hold, the only requirement on $M$ is then that all elements are non-negative. We refer to this latter situation as ``general $M$'', as opposed to $M$ satisfying detailed balance.

The vector-matrix final-size equation (\ref{eq:finalsize-vector}) admits a convenient reparameterization when only the column sums or the row sums of \(M\) are known. We factorize \(M\) so that the known information is collected in a diagonal matrix and the unknown mixing is captured by a stochastic matrix. If the column sums \(\boldsymbol{c}=(c_1,\dots,c_k)\) are known, we write
\[
M \;=\; P\,D_{\boldsymbol{c}},\qquad P\ge 0,\ \ P^{\!\top}\one=\one,
\]
where \(D_{\boldsymbol{c}}\) collects the susceptibility load of each  type, and the $j^{\rm th}$ column \(p_{\cdot j}=\sum_{i=1}^k p_{ij}\) attributes type-\(j\) infections back to the sources. Substitution into \eqref{eq:finalsize-vector} gives 
\begin{equation*}\label{eq:A-column}
    \boldsymbol{1}-\boldsymbol{\tau} = \rm{e}^{-D_{\boldsymbol{\pi}}^{-1}\,D_{\boldsymbol{c}}\,P^{\!\top} D_{\boldsymbol{\pi}}\;\boldsymbol{\tau}}.
\end{equation*}
If the row totals \(\boldsymbol{r}=(r_1,\dots,r_k)\) are given, we write
\[
M \;=\; D_{\boldsymbol{r}}\,Q,\qquad Q\ge 0,\ \ Q\one=\one,
\]
where \(D_{\boldsymbol{r}}\) contains the infectious output of each donor type and the \(i\)-th row \(q_{i\cdot}=\sum_{j=1}^k q_{ij}\) distributes transmissions caused by type \(i\) across targets. This yields
\begin{equation*}\label{eq:A-row}
    \boldsymbol{1}-\boldsymbol{\tau} = \rm{e}^{-D_{\boldsymbol{\pi}}^{-1}\,Q^{\!\top} D_{\boldsymbol{r}}\, D_{\boldsymbol{\pi}}\;\boldsymbol{\tau}} .
\end{equation*}

\section{Main results}\label{Sec-Main}
In the current section, we present our results on upper and lower bounds for the basic reproduction number $R_0$, the final size vector $\boldsymbol{\tau}=(\tau_1,\dots , \tau_k)$, and the total final size $\bar{\tau}=\sum_{i=1}^{k}\pi_i\tau_i$, where $\{\pi_i\}$ are the known community fractions of the different types, when all that is known about $M$ are the row or column sums. We consider both a general next-generation matrix $M$, but also the case where $M$ is restricted to satisfy detailed balance. 

We begin by deriving bounds for $R_0$, and then turn to bounds for both $\boldsymbol{\tau}$ and $\bar{\tau}=\sum_{i=1}^{k}\pi_i\tau_i$, noting that some cases remain unresolved. In Table \ref{Tab-summary} we describe where results for the different situations appear. 
\begin{table}[h!]
\centering
    \caption{The table describes the parameter bounds under different information scenarios}
    \bigskip
	\begin{tabular}{| c | c | c |}
		\hline
		Parameter & General $M$ & Detailed balance $M$ 
		\\ \hline
$R_0$&Theorem \ref{thm:R0-rowcol} &Theorem \ref{thm:low-cols-db}
\\ \hline
 $(\tau_1,\dots , \tau_k)$
 & $\{c_j\}$: Theorem \ref{thm:fs-col}
 & $\{c_j\}$: Theorem \ref{thm:DBk=2} ($k=2$), $k>2$ open
\\
& $\{r_i\}$: Theorem \ref{thm:fs-row}
& $\{r_i\}$: Theorem \ref{thm:DBk=2} ($k=2$), $k>2$ open
\\
\hline
$\bar\tau =\displaystyle\sum_{i=1}^k\pi_i \tau_i$
& $\{c_j\}$: Theorem \ref{thm:fs-col}
& $\{c_j\}$: Theorem \ref{thm:DBk=2} ($k=2$), $k>2$ open
\\
& $\{r_i\}$: Theorem \ref{thm:tot-fs-row}
& $\{r_i\}$: Theorem \ref{thm:DBk=2} ($k=2$), $k>2$ open
\\
\hline
	\end{tabular}
    \label{Tab-summary}
\end{table}

Throughout, a bound is said to be sharp if there exists an admissible next-generation matrix $M$ that yields a target quantity ($R_0$, $\tau_i$ or $\bar{\tau}$) that comes arbitrary close to the bound. The admissible class may include reducible $M$s.
If the extremizer, $M$ say, for a given bound is irreducible, then the bound is attained at that $M$.  If the extremizer $M$ is reducible, so some entries are necessarily zero, we may replace those entries by a small 
$\varepsilon>0$ (so that the matrix becomes irreducible), and the outcome (e.g.~$\bar{\tau}$) then tends to the bound as $\varepsilon$ tends to $0$; see Lemma \ref{lem:cont-eps} for more details. Finally, we say that bounds for different types are simultaneously sharp if there exists a single matrix $M$ so that the bound is attained for each type, with the above interpretation if $M$ is reducible.

%Throughout, a bound is said to be sharp if there are next-generation admissible matrices $M$ that come arbitrary close to the bound of the target quantity ($R_0$, $\boldsymbol{\tau}$, or $\bar{\tau}$). The admissible class may include reducible $M$s, in fact whenever an extremizer would be reducible because some entries are zero, replace those zero entries by a small $\varepsilon>0$ (so that the matrix becomes irreducible), compute the corresponding outcome and then let $\varepsilon$ tend to $0$, see Lemma \ref{lem:cont-eps} for more details. Finally, when we say that bounds are simultaneously sharp, we mean that the matrices attaining the smallest/largest outcome for one specific type need be the same matrices attaining the smallest/largest outcome for all other types, so the extrema is achieved by a single $M$ at the same time.

\subsection{Bounds for $R_0$}

\subsubsection{General $M$}
The following theorem is based on a well known linear algebra result, see \cite{CW}.
\begin{theorem}\label{thm:R0-rowcol}
Depending on whether row or columns sums of $M$ are known, the sharp bounds on $R_0$ are given by
    \begin{align*}
    &\quad &r_{\min}\;\le\; &R_0 \;\le\; r_{\max},\\
    &\quad &c_{\min} \;\le\; &R_0 \;\le\;  c_{\max}.
    \end{align*}
Further, if both the row and column sums are known, then
\[
\max \{ r_{\min},\, c_{\min} \} \;\le\; R_0 \;\le\; \min \{ r_{\max},\, c_{\max} \}.
\]
These bounds are not necessarily sharp.
\end{theorem}
It follows immediately from Theorem~\ref{thm:R0-rowcol} that the uncertainty of $R_0$ reduces the closer to each other all column sums are, and similarly for row sums.

\subsubsection{Detailed balance}
\begin{theorem}\label{thm:low-cols-db}
Depending on whether row and/or columns sums of $M$ (satisfying detailed balance) are known, the bounds on $R_0$ are given by
\begin{align*}
    &\quad &\bar{r}\;\;\le\;\; &R_0 \;\le\;\; r_{\max},\\
    &\quad & \tilde{c}  \;\;\le\;\; &R_0 \;\le\;\; c_{\max},\\
    &\quad &\max \{\bar{r}, \tilde{c} \} \;\le\; &R_0 \;\le\; \min \{ r_{\max},\, c_{\max} \},
    \end{align*}
where
\begin{equation*}
    \bar{r}=\sqrt{\sum_{i=1}^k \pi_i r_i^2}\qquad\textit{and}\qquad\tilde{c}=\sqrt{\frac{\sum_{j=1}^k c_j^2/\pi_j}{\sum_{i=1}^k 1/\pi_i}}
\end{equation*}
Only the upper bounds derived from knowing exclusively the row sums or exclusively the column sums are necessarily sharp.
\end{theorem}
In Section \ref{recipe} we describe a numerical procedure to compute sharp lower bounds for any number of types $k$. For the special case $k=2$ with known row sums, an explicit formula for the sharp lower bound is given in Theorem \ref{thm:DBk=2}.

\subsection{Bounds for $\{\tau_i\}$ and $\bar\tau$}
We now derive bounds for the final-size vector $\boldsymbol{\tau}$, defined as the solution of~\eqref{eq:finalsize-vector}, and for the corresponding total final size $\bar{\tau}$.

\subsubsection{General $M$}
\begin{theorem}[Known column sums]\label{thm:fs-col}
For each $i$,
\begin{equation}\label{eq:col-component1}
1-\exp\!\Big(-\frac{c_i}{\pi_i}\,y_*\Big)
\;\le\;
\tau_i
\;\le\;
1-\exp\!\Big(-\frac{c_i}{\pi_i}\,y^*\Big).
\end{equation}
where 
\begin{equation*}
\label{eq:ystar}
    y_*=\min_j\{\pi_j\;t_{c_j}\},
\qquad
y^*=\max_j\{\pi_j\;t_{c_j}\},
\end{equation*}
and $t_{\alpha}$ is defined in (\ref{t-alpha}). The bounds are sharp componentwise but also simultaneously. As a direct consequence, the overall fraction infected $\bar\tau=\sum_i\pi_i\tau_i$ has sharp bounds given by
\[
1-\sum_{i=1}^k \pi_i \exp\!\Big(-\frac{c_i}{\pi_i}\,y_*\Big)
\;\le\;
\bar{\tau}
\;\le\;
1-\sum_{i=1}^k \pi_i \exp\!\Big(-\frac{c_i}{\pi_i}\,y^*\Big).
\]
\end{theorem}
Theorem~\ref{thm:fs-col} says that each type's final size is contained between two extreme situations. The lower extreme is where all infectivity comes from the type with the lowest per capita exposure, and the upper limit arises when all infectivity comes from the type with the highest per capita exposure.
These interpretations follow from the explicit extremal choices of $M$, see Section \ref{subsec:proof_R0_DB}.

\begin{remark}\label{remark:trivial}
Note that if $c_j\le 1$ for some $j$, then $t_{c_j}=0$, whence $y_*=0$ and the lower bound in~\eqref{eq:col-component1} reduces to $0$ %$0\le \tau_i$
for all $i$. In other words, the lower bound becomes trivial as soon as there exists a type whose mean number of infections received from the whole community is at most one.
Moreover, if all types are equally frequent, i.e.\ $\pi_j= 1/k$ for all $j$, then
$y_*=\frac{1}{k}t_{c_{\min}}$ and $y^*=\frac{1}{k}t_{c_{\max}}$ depend only on the extreme column totals. 
\end{remark}

We now consider the situation where instead of column sums, the row sums $\{r_i\}$ are observed. For this we need to introduce $t_{\alpha,\gamma}$, $\kappa(x)$ and $K_i$. For $\alpha\ge 0$ and $\gamma\ge 0$, let 
\begin{equation}\label{eq:t_alpha-gamma}
t_{\alpha,\gamma}\in[0,1]\quad\text{denote the maximal solution of}\quad t=1-\exp\left(-\alpha t - \gamma\right).
\end{equation}
The above equation can be interpreted as the homogeneous final size equation with basic reproduction number $R_0=\alpha$, and an added constant $\gamma$ representing a cumulative external force of infection. Thus $t_{\alpha,\gamma}$ is the attack rate consistent with intrinsic transmission $\alpha$ and external pressure $\gamma$. We note that if there is no external force of infection ($\gamma =0$) then $t_{\alpha,0}=t_\alpha$.
Let
\begin{equation}\label{eq:kappa}
\kappa(x)\;:=\;
\begin{cases}
0, &\quad if\quad x\le 1,\\[2pt]
x-1-\log x, &\quad if\quad x>1,
\end{cases}
\end{equation}
implying that $\kappa(x)\ge 0$, and set 
\begin{equation}\label{eq:K-i}
    K_i \;:=\; \dfrac{1}{\pi_i}\sum_{j\neq i}\pi_j\,\kappa(r_j)\,.
\end{equation}
Given a general next-generation matrix $M$ with known row sums, we have the following result. 
\begin{theorem}[Known row sums]\label{thm:fs-row}
For each $i$,
\begin{equation}
    0 \;\le\; \tau_i \;\le\; t_{r_i,K_i},
\end{equation}
Moreover, the bounds are sharp componentwise but not simultaneously.
\end{theorem}
The value $\kappa(r_j)$ represents an upper bound on the imported force of infection of the type $j$ that can send to others after keeping the minimum needed to sustain its own transmission. 
\begin{remark}
The lower bound is simply obtained with a next-generation matrix $M$ that has the $i^{th}$ column null. The upper bound is obtained when every other type $j \neq i$ directs all of its exportable infections to $i$, while keeping only the minimum needed to sustain its within-type final size.
\end{remark}
Since the bounds in Theorem \ref{thm:fs-row} are not simultaneously sharp, the weighted sum of the separate lower and upper terms do not yield sharp lower and upper bounds for the overall fraction infected, i.e.~the total final size. Instead, we have the following result. 
\begin{theorem}[Known row sums]\label{thm:tot-fs-row}
Let $\bar{\tau}=\sum_{i=1}^k \pi_i\,\tau_i$. Then,
\begin{equation}
\min_{i}\, \bigl\{\pi_i\, t_{r_i}\bigr\}
\;\le\; 
\bar{\tau}
\;\le\; 
\tau^\star,
\end{equation}
where 
\[
\tau^\star=
\begin{cases}
    0& if\quad r_i\leq 1\quad \forall\,i\,\in\{1,...,k\}\\
    \displaystyle\sum_{i=1}^k \pi_i\Bigl[\ 1-\frac{\lambda^\star}{1+\lambda^\star r_i}\ \Bigr]_+& if\quad \exists\,i\,\in\{1,...,k\}\quad s.t.\quad r_i>1
\end{cases}
\]
and $\lambda^\star>0$ is the unique positive solution of
\begin{equation}
\label{eq:lambdastarequ}
\sum_{i=1}^k \pi_i\Bigl[\log\!\frac{1+\lambda r_i}{\lambda}\ \Bigr]_+
\;=\;
\sum_{i=1}^k r_i\,\pi_i \Bigl[\ 1-\frac{\lambda}{1+\lambda r_i}\ \Bigr]_+,
\end{equation}
where $[\cdot]_+$ is the usual positive part function (i.e. $[x]_+=x$ if $x\geq 0$ and $[x]_+=0$ otherwise).
Moreover, the bounds are sharp. 
\end{theorem}

The upper bound is achieved by a rank-one (separable) mixing matrix
(see the proof in Section \ref{subsubsectio.proof.thm:tot-fs-row}). In this case
all infectious individuals send their contacts to types according to the same
probability vector $\boldsymbol{q}=\boldsymbol{u}^\star/\sum_{i=1}^k{u_i}^\star\propto \boldsymbol{u}^\star$, where
\[
u^\star_i=\pi_i\Bigl[\log\!\frac{1+\lambda^\star r_i}{\lambda^\star}\ \Bigr]_+
\]
can be interpreted as the incoming infectious pressure on type $i$, i.e. the cumulative infectious hazard received by type $i$ over the course of the epidemic. It means that most infectious contacts go to the
types where they are most likely to generate new infections. Final sizes,
however, depend only on the received infection pressures $\boldsymbol{u}^\star$, not on which
types send them: any mixing pattern that produces the same $\boldsymbol{u}^\star$ yields the
same value of $\bar{\tau}$, see Lemma \ref{lem:familyofQ}. 
\begin{remark}
If all types have the same row sum $r_i=R_0>1$, then $u_i^\star$ is proportional to the type sizes $\pi_i$.
Consequently, the one of the optimal targeting distribution satisfies $\boldsymbol{q}\propto\boldsymbol{u}^\star\propto\boldsymbol{\pi}$, hence $\boldsymbol{q}=\boldsymbol{\pi}$. In this case the upper bound is attained by the rank-one mixing matrix $M=R_0\,\boldsymbol{1}\boldsymbol{\pi}^\top$, meaning that each infectious individual generates $R_0$ secondary infections and allocates them across types proportionally to their population shares.
The upper bound then simplifies to a common final size
\[
\bar\tau^\star = 1-\frac{\lambda^\star}{1+\lambda^\star R_0}.
\]
Moreover, \eqref{eq:lambdastarequ} implies
\[
\log\!\frac{1+\lambda^\star R_0}{\lambda^\star}= 1-\frac{\lambda^\star}{1+\lambda^\star R_0}.
\]
Substituting $\frac{\lambda^\star}{1+\lambda^\star R_0}=1-\bar{\tau}^{\star}$ and rearranging yields $=1-\bar{\tau}^{\star}={\rm e}^{-R_0\bar{\tau}^{\star}}$, whence, $\bar{\tau}^{\star}=t_{R_0}$.
\end{remark}

\begin{remark}
Note that here each term in the sum $\pi_i\!\left[1-\frac{\lambda^\star}{1+\lambda^\star r_i}\right]_+$ can be interpreted as the maximum contribution of type $i$ to the total final size under the mixing pattern constructed in the proof. Therefore, if there is a rare type with a large $r_i$, its contribution
\[
\pi_i\!\left[1-\frac{\lambda^\star}{1+\lambda^\star r_i}\right]
\]
can be close to $\pi_i$ (corresponding to an almost complete outbreak within that type), even when $\pi_i$ is small.
\end{remark}

\subsubsection{Detailed balance}
Finally, we consider bounds on $\{\tau_i\}$ and $\bar\tau$ given the row (or column) sums when the next-generation matrix $M$ satisfies detailed balance. 
Of course, the corresponding bounds for a general $M$ in Theorems~\ref{thm:fs-col}-\ref{thm:tot-fs-row} above continue to hold, but typically they are no longer sharp under the extra constraint imposed by detailed balance.

Substituting the detailed-balance condition $\pi_{i}m_{ij}=\pi_j m_{ji}$ into the final-size  equations~\eqref{eq:finalsize} yields
\begin{align}
\label{eq:finalsizeDB}
1-\tau_j &= {\rm e}^{-\sum_{i=1}^k \tau_i m_{ji}}  ,\quad j=1,\dots ,k.
\end{align}
\begin{remark}
Suppose that all types have the same row sum $r_i=R_0$.  Then it is easily seen that $\tau_j=t_{R_0}$, $j=1,\dots,k$, satisfies~\eqref{eq:finalsizeDB}.  Thus, under detailed balance, the final size is fully determined and $\bar{\tau}=t_{R_0}$ if the row sums all equal $R_0$. 
\end{remark}

Under detailed balance with $k$ types $M$ has $k(k+1)/2$ free parameters, and if the $k$ row sums are given, there remain $k(k-1)/2$ free parameters. Hence, the case $k=2$ has one free parameter, making analytic progress easier, and we give a fairly complete analysis of this case with fixed row sums in Theorem~\ref{thm:DBk=2} below.  (A broadly similar analysis is available for the case $k=2$ with fixed column sums, though the details are more complicated.  We omit this case as it is of less practical relevance.)  

For $k>2$ the number of free parameters is $k(k-1)/2 \ge 3$ so the analysis is much harder, for this situation we only have a partial result and a general conjecture.  
\label{subsec:DBM}

\medskip
\textbf{The case} $\boldsymbol{k=2}.$\\
Suppose the row sums $r_1$ and $r_2$ are fixed, $r_1<r_2$, $\pi_1 \in (0,1)$, and let $\pi_2=1-\pi_1$.  The next-generation matrix takes the form
\[
M=M(\theta)=\begin{bmatrix}r_1-\theta & \theta \\ \varphi\,\theta  & r_2-\varphi\,\theta\end{bmatrix}, \qquad 0 \le \theta \le \theta_{\max}=\min(r_1,\varphi^{-1}r_2),
\]
where $\varphi =\frac{\pi_1}{\pi_2}\in(0, \infty)$. For $\theta \in [0, \theta_{\max}]$, define $R_0(\theta), \tau_1(\theta), \tau_2(\theta)$ and $\bar{\tau}(\theta)$ in the obvious fashion.  We show in Theorem \ref{thm:DBk=2}(a) below that $R_0(\theta)$
is strictly decreasing.  If $R_0(0)=r_2 \le 1$, then $\tau_1(\theta) =\tau_2(\theta)=\bar{\tau}(\theta)=0$ for all $\theta \in [0, \theta_{\max}]$.  In the following we assume that $r_2>1$.  If $R_0(\theta_{\max}) < 1$, let $\theta_*$ be the unique solution of $R_0(\theta)=1$, otherwise let $\theta_*=\theta_{\max}$.  Note that if $R_0(\theta_{\max}) \le 1$ then $\tau_1(\theta) =\tau_2(\theta)=\bar{\tau}(\theta)=0$ for all $\theta \in [\theta_*, \theta_{\max}]$.

\medskip
Before stating the main result, some more notation is required. Let
\begin{equation}
\label{equ:fidef}
f_i(\theta)=\frac{1}{1-\tau_i(\theta)}-r_i,\qquad\theta \in [0, \theta_{\max}],\quad  i=1,2.
\end{equation}
Let
\[
g(r)=\frac{1}{1-t_r}-r, \qquad r \in (0, \infty).
\]
Thus, $g(r)=1-r$ for $r \in (0,1]$.  We show in Section \ref{subsec:thm:DBk=2:proof} that $g$ is strictly increasing on $[1, \infty)$ and that for $r \in (0,1)$, there exists a unique $\hat{r}_2(r) \in (1, \infty)$ such that $g(\hat{r}_2(r))=g(r)$.  Let $A=\{(r_1, r_2) \in \mathbb{R}_+^2:0<r_1<1, 1 <r_2 \le \hat{r}_2(r_1)\}$.  We show in Section \ref{subsec:thm:DBk=2:proof} that if $(r_1, r_2) \in A^c$ and $f_1(\theta_*)>f_2(\theta_*)$, then there is a unique $\theta_1 \in (0, \theta_*)$ such that $f_1(\theta_1)=f_2(\theta_1)$.
\begin{theorem}
\label{thm:DBk=2}
Assume detailed balance, $k=2$, and $r_2>1$ for parts (b) and (c). 
\begin{itemize} 
\item[(a)] $R_0(\theta)$ is strictly decreasing on $[0, \theta_{\max}]$.  Further,
\begin{equation}
\label{equ:DBk=2R0boundsa}
\frac{1}{2}\left[r_2-\varphi r_1+\sqrt{(r_2-\varphi r_1)^2+4 \varphi r_1^2}\right] \le R_0 \le r_2 \quad \text{if} \quad\varphi r_1 \le r_2,
\end{equation}
and
\begin{equation}
\label{equ:DBk=2R0boundsb}
\frac{1}{2}\varphi^{-1}\left[\varphi r_1 -r_2 +\sqrt{(\varphi r_1-r_2)^2+4 \varphi r_2^2}\right] \le R_0 \le r_2 \quad \text{if} \quad\varphi r_1 \ge r_2.
\end{equation}
Moreover, these bounds are sharp.
\item[(b)]
\begin{itemize}
\item[(i)] $\tau_2(\theta)$ is strictly decreasing on $[0,\theta_*]$. 
\item[(ii)] $\tau_1(\theta)$ is strictly increasing on $[0,\theta_{\max}]$ if $f_2(\theta_{\max}) \ge 0$.  Otherwise, there exists a unique $\theta_0 \in (0, \theta_*)$ such that $f_2(\theta_0)=0$, and $\tau_1(\theta)$ has a unique turning point in $(0, \theta_*)$, which is a maximum, at $\theta_0$.
\end{itemize}
\item[(c)]
\begin{itemize}
\item[(i)] Suppose $(r_1, r_2) \in A$.  Then $\bar{\tau}(\theta)$ is strictly decreasing on $[0,\theta_*]$. 
\item[(ii)] Suppose $(r_1, r_2) \in A^c$.  Then $\bar{\tau}(\theta)$ is strictly increasing on $[0,\theta_{\max}]$ if $f_1(\theta_*) \le  f_2(\theta_*)$.  Otherwise, $\bar{\tau}(\theta)$ has a unique turning point in $(0, \theta_*)$, which is a maximum, at $\theta_1$.
\item[(iii)]  Let $\hat{\theta} = argmax(\bar{\tau}(\theta))$.  If $\hat{\theta} \in [0, \theta_*)$, or $\hat{\theta} =  
\theta_*$ and $f_1(\theta_*)=f_2(\theta_*)$, then $\bar{\tau}(\hat{\theta})$ attains the upper bound for $\bar{\tau}$ given in Theorem \ref{thm:tot-fs-row}.
\end{itemize}
\end{itemize}
\end{theorem}
\begin{remark}
Sharp bounds for $\tau_1, \tau_2$ and $\bar{\tau}$ are obtained easily from Theorem \ref{thm:DBk=2}.  In particular, $\min(\bar{\tau}(0), \bar{\tau}(\theta_{\max}))$ is a sharp lower bound for $\bar{\tau}$; and  $\max(\bar{\tau}(0), \bar{\tau}(\theta_{\max}))$ is a sharp upper bound for $\bar{\tau}$ unless $(r_1, r_2)$ belongs to the region where $\bar{\tau}(\theta)$ has an internal maximum, in which the upper bound in Theorem \ref{thm:tot-fs-row} is also sharp under detailed balance ($k=2$).
\end{remark}

\begin{figure}[H]
    \centering
    \includegraphics[scale=0.65]{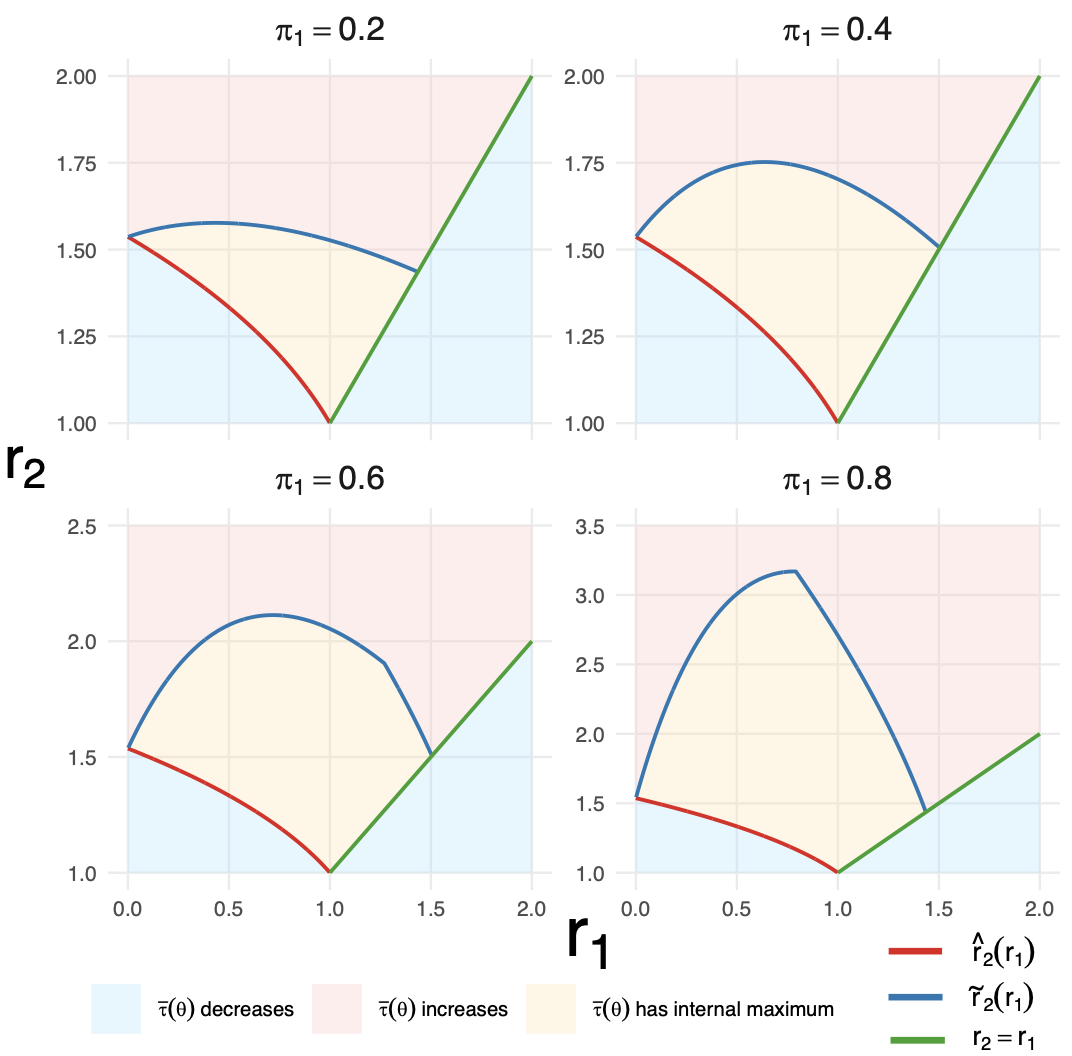}
    \caption{Plots showing regions in the $(r_1,r_2)$ plane of different behaviours of $\bar{\tau}(\theta)$ for different values of $\pi_1$.  The solid red curve is $\hat{r}_2(r_1)$ and the blue curve is $\tilde{r}_2(r_1)$.  The green line is $r_2=r_1$.  Recall that $\bar{\tau}(\theta)$ is identically zero if $r_2 \le 1$ and the analysis assumes $r_2>r_1$.}
\label{fig:2}
\end{figure}

A shortcoming in Theorem \ref{thm:DBk=2} of (c)(ii) is that it does not indicate for which $(r_1, r_2)$ we have $f_1(\theta_*) \le f_2(\theta_*)$.  Numerical investigation shows that there exists $r_1^*=r_1^*(\varphi)$ such that if $r_1 \ge r_1^*$, then $f_1(\theta_*) \le f_2(\theta_*)$ for all $r_2>r_1$.  Further, if $r_1 < r_1^*$, then there exists $\tilde{r}_2(r_1)$, such that $f_1(\theta_*) > f_2(\theta_*)$ if $r_2 \in (\hat{r}_2(r_1), \tilde{r}_2(r_1))$ and $f_1(\theta_*) \le f_2(\theta_*)$ if $r_2 \ge \tilde{r}_2(r_1)$.

Plots showing regions in the $(r_1,r_2)$ plane of different behaviors of $\bar{\tau}(\theta)$ for different values of $\pi_1$ are shown in Figure \ref{fig:2}.  (Note that the scale on the $y-$axis differs between some subplots.)  The kink in the graphs of $\tilde{r}_2(r_1)$ when $\pi_1>1/2$ is when $\theta_{\max}=\min(r_1, \varphi^{-1}r_2)$ changes from $r_1$ to $\varphi^{-1}r_2$.

\medskip
\textbf{The case} $\boldsymbol{k>2}.$\\
\label{subsubsec:DBMk>2}
Let $\mathcal{M}_{\boldsymbol{\pi},\boldsymbol{r}}^{\rm DB}$ be the set of all next-generation matrices $M$ that satisfy detailed balance and are consistent with $\boldsymbol{\pi}$ and $\boldsymbol{r}$. We will suppress the dependence of $\boldsymbol{\tau}$, $\bar{\tau}$, and $f_1,f_2,\dots,f_k$ on $M$ to ease the notation.
\begin{theorem}
\label{thm:DBk>2}
Suppose that there exists $M \in \mathcal{M}_{\boldsymbol{\pi},\boldsymbol{r}}^{\rm DB}$, $A \subseteq \{1,2,\dots,k\}$ and $c \in(0,\infty)$ such that
\begin{itemize}
 \item[(i)]    $\tau_i >0$ and $f_i=c$ for all $i \in A$;
 \item[(ii)]  $\tau_i=0$ and $c+r_i\le 1$ for all $i \notin A$.
 \end{itemize}
Then $\bar{\tau}$ attains the upper bound for $\bar{\tau}$ given in Theorem \ref{thm:tot-fs-row}.
\end{theorem}

The following conjecture, a converse to Theorem~\ref{thm:DBk>2}, is based on numerical evidence.  It holds when $k=2$ (see (c)(ii) in Theorem \ref{thm:DBk=2}) and a similar argument shows that it holds when 
$k>2$ and $A$ has cardinality 2.

\begin{conjecture}\label{conj:DBk>2}
For $A \subseteq \{1,2,\dots,k\}$, let $\mathcal{M}_{\boldsymbol{\pi},\boldsymbol{r}}^{\rm DB}(A)$ be the set of $M \in \mathcal{M}_{\boldsymbol{\pi},\boldsymbol{r}}^{\rm DB}$ satisfying $m_{ij}=0$ for all $i, j \notin A$. Suppose that $M={\rm argmax}_{\mathcal{M}_{\boldsymbol{\pi},\boldsymbol{r}}^{\rm DB}}\bar{\tau}$ 
belongs to the interior of $\mathcal{M}_{\boldsymbol{\pi},\boldsymbol{r}}^{\rm DB}(A)$ for some $A  \subseteq \{1,2,\dots,k\}$.  
Then (i) and (ii) in Theorem~\ref{thm:DBk>2} hold and $\bar{\tau}$ attains the upper bound for $\bar{\tau}$ given in Theorem \ref{thm:tot-fs-row}.
\end{conjecture}

\section{Numerical Illustrations}\label{Sec-Illustr}
In the current section we illustrate our findings by showing the range of possible $R_0$ and final sizes $\bar \tau$ for given row sums in a simple 2-type example, and for a Belgian social contact study where individuals are classified according to age as well as social activity, and where the row sums are known but where the distinct elements of $M$ are unknown.

\subsection{An example of bounds for $R_0$ and final size $\bar\tau$ when $k=2$}\label{Bounds-k-2}
We start by presenting the bounds of $R_0$ as well as the overall final size $\bar\tau$ numerically, when there are only two types of individuals with population frequencies $\pi_1$ and $\pi_2=1-\pi_1$. We assume that we only observe the row sums of $M$: $r_1$ and $r_2$, the average number of infectious contacts that each type has, respectively. In Figure \ref{fig:tauboundr} we show the bounds for $R_0$ and $¶\bar \tau$ when $\pi_1=0.2,\ \pi_2=0.8,\ r_1=1.6$ and varying $r_2$. We consider both the case where the next-generation matrix $M$ is of arbitrary form, and when it satisfies detailed balance. 

For $R_0$ we give the bounds valid for general $k$ (the upper being sharp but the lower not), but we also include the sharp lower bound from Theorem \ref{thm:DBk=2} valid only when $k=2$. As for the overall final size $\bar \tau$, results are only available for $k=2$ so these are the ones plotted.

\begin{figure}[H]
\resizebox{13.25cm}{!}{\includegraphics{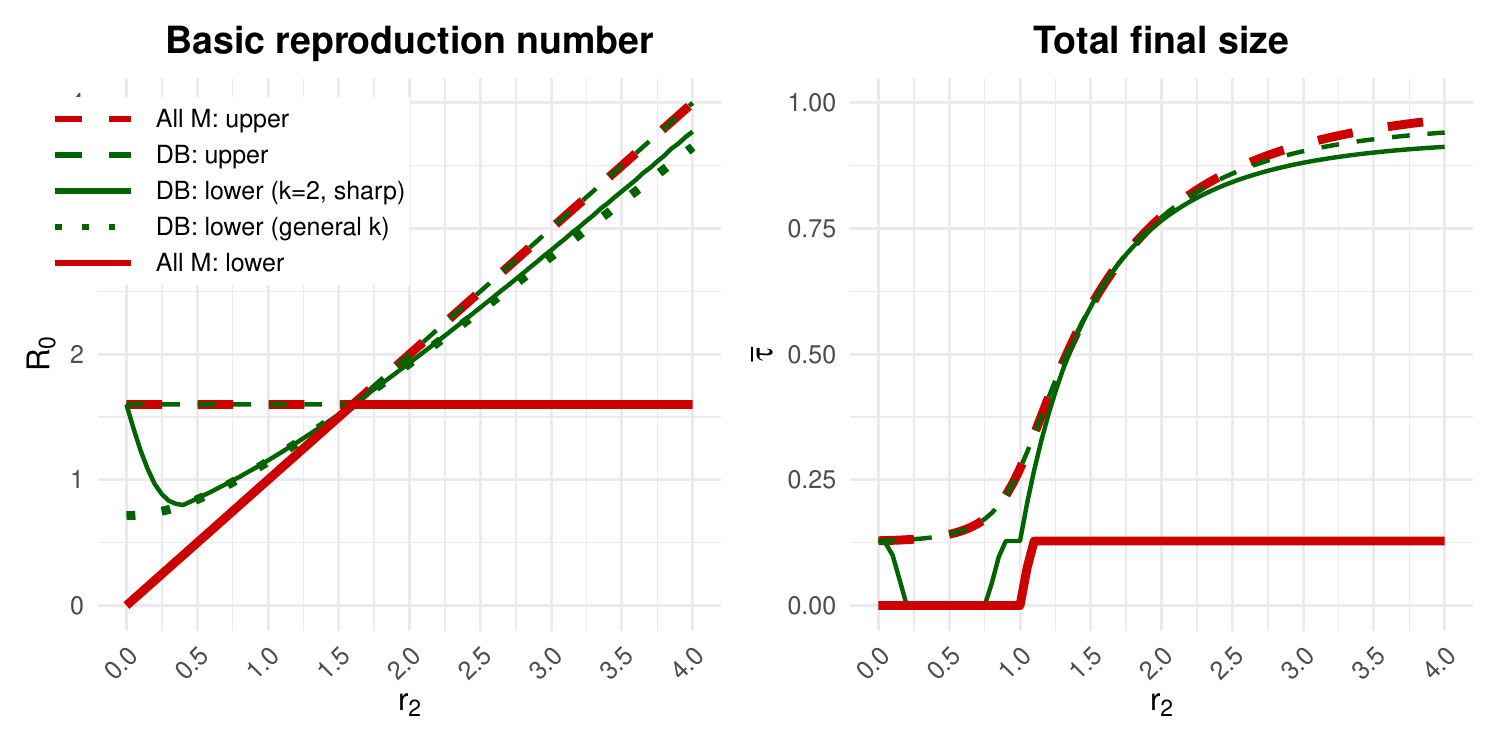}}
\caption{Plots showing lower and upper bounds for $R_0$ and $\bar{\tau}$, both under generic $M$ and under detailed balance restriction, knowing only the row sums, when there are $k=2$ types of individuals, $\pi_1=0.2$ and $r_1=1.6$.}
\label{fig:tauboundr}
\end{figure}
We note that, for large $r_2$ (highly super-critical case) the upper and lower bounds are very similar if $M$ satisfies detailed balance (for both $R_0$ and $\bar\tau$), suggesting that under detailed balance the row sums alone provide fairly accurate estimates of these quantities. However, if detailed balance does not hold, the lower bound is much lower than the upper (the lower bound is achieved when both types infect only type 1 individuals). 

Perhaps more surprising is the fact that the lower bound for both $R_0$ and $\bar\tau$ under detailed balance is actually \emph{decreasing} in $r_2$ for small $r_2$: if $r_2=0.5$ the lower bound is actually \emph{smaller} than if $r_2=0$. An explanation of this is that when $r_2=0$, all contacts that 1-individuals make must be with other 1-individuals, making type 1 super critical on its own. But when $r_2=0.5$, some of the contacts 1-individuals make can be with 2-individuals making type 1 sub-critical, with type 2 remaining sub-critical (note that type 2 makes up 80\% of the community).

The green dotted curve in the left panel represents the general lower bound under the detailed balance assumption given in Theorem~\ref{thm:low-cols-db}. It is obtained as a simple square root of the weighted average of the square of the row sums (with weights given by the population fractions), hence it is available in closed form and can be computed for any $k$. Although this bound is not sharp in general, it seems to provide a surprisingly good approximation of the sharp lower bound, except for small values of $r_2$.

\subsection{Bounds for $R_0$ and final size $\bar\tau$ for a Belgian Social contact study}\label{Bounds-empirical}

In \cite{BB25}, a Belgian social contact study (\cite{WIL}) is reanalysed. It is shown that social activity plays a greater role for the potential of epidemic outbreaks than the age-classes do. To simplify our illustration we here merge age classes to only two: up to 18 years old and 19+ years, comprising 21.50\% and 78.50\% of the population, respectively. The age-specific contact matrix (obeying detailed balance) is given by 
\begin{equation}
\label{equ:C2by2}
C=\begin{bmatrix} 10.6560& 9.0423\\2.4766 & 15.4571\end{bmatrix},
\end{equation}
where the elements reflect mean number of contacts per day between the two types. It was shown in \cite{BB25} that the variation in number of contacts was substantial within age-groups. To capture this variation each age-group is divided into two halves: the socially active and the socially less active, thus creating four types of individuals: $CS, CL, AS, AL$ children and adults being either social or less social. A problem with this improved model is however that we no longer observe all elements in the extended contact matrix $C$, which is now $4\times 4$. In each age-group we know who are the socially active and who are less socially active, and we know the mean number of contacts with the two age-groups each of them have. But, we do not know if the socially active group have most contacts with other socially active in the corresponding age-group (referred to as assortative mixing with respect to social activity), or if they have many of their contacts with less socially active people (referred to as disassortative). Consequently, we are in the situation that we know all four row sums of $C$, but not the separate elements of $C$. 
Assuming a mean infectious period of one day, the next-generation matrix is therefore written as $M=pC$, where $p$ is the transmission probability (for other mean infectious periods the scale should be multiplied accordingly). 

In Figure \ref{fig:Belgian-bound} we have plotted the lower and upper bounds of $R_0$ and the overall final epidemic size $\bar\tau$ (using the results of Section \ref{Sec-Main}) as functions of the transmission probability $p$, for fixed row sums making no assumption on $M$ (red) and assuming $M$ satisfies detailed balance (green). The third (yellow) bounds are obtained by taking into account that we also know the sums of some pairs of elements of $M$: the mean number of contacts made by any given type with any given age group.
Further details concerning numerical calculation of the bounds may be found in Appendix~\ref{secA2}.
%we here also know pairs of elements of $M$: the number contacts with socially as well as less active individuals of between each age classes.

\begin{figure}[H]
\resizebox{14cm}{!}{\includegraphics{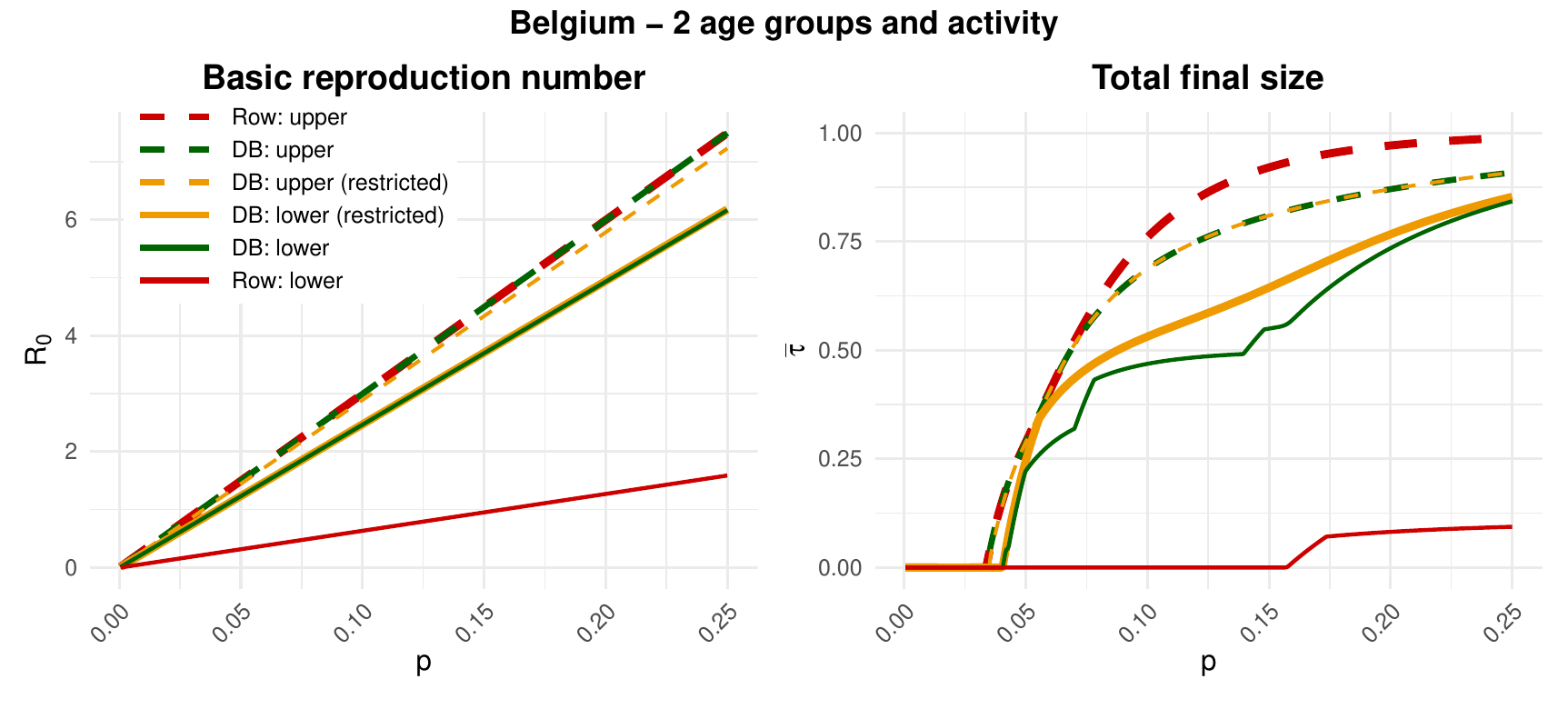}}
\caption{Different bounds for $R_0$ and the final size $\bar\tau$ in the Belgian study taking heterogeneity within age-groups into account.}
\label{fig:Belgian-bound}
\end{figure}

In the figure it is seen that the range of possible outcomes is quite large, suggesting that obtaining knowledge about degree of assortativity with respect to social activity will improve precision in conclusions. In \cite{BB25} it was shown that neglecting heterogeneity in social activity reduces the dimension of $M$ with all elements of $M$ known, this giving exact expressions for $R_0$ and $\boldsymbol{\tau}$, but instead the curves are erroneous when in fact there is heterogeneity within age-groups.

%Figure \ref{fig:Belgian-bound} shows plots of the bounds for $R_0$ and $\bar\tau$ as a function of the transmission probability $p$. The bounds labelled `Row: upper' and `Row: lower' are obtained using Theorem \ref{thm:R0-rowcol} for $R_0$ and Theorem~\ref{thm:tot-fs-row} for $\bar{\tau}$, given the four row sums of $M$.  Those labelled `DB: upper' and `DB: lower' assume also that $M$ satisfies detailed balance.  Those labelled `DB: upper (restricted)' and `DB: upper (restricted)' assume that for each of the four types of individuals, the average number of contacts made with children and with adults are both known.

We observe also that the bounds tighten as more information is incorporated, since additional constraints shrink the set of admissible $M$; this effect is particularly notable when imposing the detailed balance condition. Moreover, since $M=pC$, the spectral radius ($R_0$) is linear in $p$, as shown in the left panel. By contrast, for the final epidemic size, the detailed balance lower bound exhibits visible kinks (i.e., changes in slope) as a function of transmission probability. This behaviour may be explained by the fact that depending on $p$, the minimising matrix $M$ within the admissible set may switch from one extremal configuration to another, corresponding to a different allocation of transmission (i.e. the row sums) across types while respecting the constraints.

\section{Proofs}\label{Sec-Proof}

\subsection{$R_0$ bounds for general $M$}
\textbf{Proof of Theorem \ref{thm:R0-rowcol}}
\begin{proof}
By Lemma~2.1 in \cite{CW},
\[
r_{\min}\ \le\ \rho(M)\ \le\ r_{\max}.
\]
Since $\rho(M)=\rho(M^\top)=R_0$, the same argument applied to $M^\top$ gives
\[
c_{\min}\ \le\ \rho(M)\ \le\ c_{\max}.
\]
Hence
\[
\max\{r_{\min},c_{\min}\}\ \le\ \rho(M)\ \le\ \min\{r_{\max},c_{\max}\},
\]
which proves the bounds.

To see that they are sharp, fix the row sums $\boldsymbol{r}$. Let $h\in\displaystyle\arg\max_{i} \{r_i\}$ and $l\in\displaystyle\arg\min_{i} \{r_i\}$, and set
\[
M^{(r,\max)}:=\boldsymbol{r}\,\boldsymbol{e}_h^\top,
\qquad
M^{(r,\min)}:=\boldsymbol{r}\,\boldsymbol{e}_l^\top,
\]
where $\boldsymbol{e}_j$ is the $j$-th canonical vector. Both matrices have row sums $\boldsymbol{r}$ and are rank-one, with unique nonzero eigenvalues
\[
\rho\!\left(M^{(r,\max)}\right)=r_{\max},
\qquad
\rho\!\left(M^{(r,\min)}\right)=r_{\min}.
\]
For fixed column sums $\boldsymbol{c}$, the lower and upper bounds are respectively attained by rank-one matrices 
\[
M^{(c,\max)}:=\boldsymbol{e}_h\,\boldsymbol{c}^\top,
\qquad
M^{(c,\min)}:=\boldsymbol{e}_l\,\boldsymbol{c}^\top.
\]
where we choose $h\in\displaystyle\arg\max_{j} \{c_j\}$ and $l\in\displaystyle\arg\min_{j}\{c_j\}$.
\end{proof}

\subsection{$R_0$ bounds for Detailed Balance}
\label{subsec:proof_R0_DB}

\subsubsection*{Proof of Theorem \ref{thm:low-cols-db}}
\begin{proof}
Let \(S=D_{\boldsymbol{\pi}}^{1/2} M D_{\boldsymbol{\pi}}^{-1/2}\). By \ref{lemma:db1}, the matrix \(S\) is symmetric; hence, $S^2$ is also symmetric. We define the spectral radius of $S^2$ by applying the Rayleigh-Ritz formula in Theorem 4.2.2 of \cite{RR}
\begin{equation}\label{eq.rr}
\rho(S^2)=\max_{\boldsymbol{x}\neq \boldsymbol{0}}\frac{\boldsymbol{x}^\top S^2 \boldsymbol{x}}{\boldsymbol{x}^\top \boldsymbol{x}}.
\end{equation}
Since \(\rho(S^2)=\rho(M^2)=R_0^2\) and because $S$ is symmetric, \eqref{eq.rr} implies that for any $\boldsymbol{x}\neq \boldsymbol{0}$ 
\begin{equation}
\label{eq_R0}
R_0\ge \sqrt{\displaystyle\frac{\boldsymbol{\boldsymbol{x}}^\top S^2 \boldsymbol{x}}{\boldsymbol{x}^\top \boldsymbol{x}}}=\frac{\|S\boldsymbol{x}\|_2}{\|\boldsymbol{x}\|_2}.
\end{equation}
Choosing \(\boldsymbol{x}=D_{\boldsymbol{\pi}}^{-1/2}\mathbf{1}\),
\[
S\boldsymbol{x}
=
D_{\boldsymbol{\pi}}^{1/2} M D_{\boldsymbol{\pi}}^{-1/2}
D_{\boldsymbol{\pi}}^{-1/2}\mathbf{1}
=
D_{\boldsymbol{\pi}}^{1/2} M D_{\boldsymbol{\pi}}^{-1}\mathbf{1}
=
D_{\boldsymbol{\pi}}^{-1/2} M^\top \mathbf{1}
=
D_{\boldsymbol{\pi}}^{-1/2}\boldsymbol{c}.
\]
Thus
\[
\|S\boldsymbol{x}\|_2
=
\|D_{\boldsymbol{\pi}}^{-1/2}\boldsymbol{c}\|_2
=
\sqrt{\sum_{i=1}^k \frac{c_i^2}{\pi_i}},
\qquad
\|\boldsymbol{x}\|_2
=
\|D_{\boldsymbol{\pi}}^{-1/2}\mathbf{1}\|_2
=
\sqrt{\sum_{i=1}^k \frac{1}{\pi_i}}.
\]
Substituting into \eqref{eq_R0} inequality gives
\[
R_0
\ge
\sqrt{
\frac{\sum_{i=1}^k c_i^2/\pi_i}{\sum_{i=1}^k 1/\pi_i}
}
=:
\tilde c.
\]
Choosing instead $\boldsymbol{x}=D_{\boldsymbol{\pi}}^{1/2}\mathbf{1},$\
\[
S\boldsymbol{x}
=
D_{\boldsymbol{\pi}}^{1/2} M D_{\boldsymbol{\pi}}^{-1/2}
D_{\boldsymbol{\pi}}^{1/2}\mathbf{1}
=
D_{\boldsymbol{\pi}}^{1/2} M \mathbf{1}
=
D_{\boldsymbol{\pi}}^{1/2}\boldsymbol{r}.
\]
Hence
\[
\|S\boldsymbol{x}\|_2
=
\|D_{\boldsymbol{\pi}}^{1/2}\boldsymbol{r}\|_2
=
\sqrt{\sum_{i=1}^k \pi_i r_i^2},
\qquad
\|\boldsymbol{x}\|_2
=
\|D_{\boldsymbol{\pi}}^{1/2}\mathbf{1}\|_2
=
\sqrt{\sum_{i=1}^k \pi_i}
=
1.
\]
Therefore substituting into \eqref{eq_R0},
\[
R_0
\ge
\sqrt{\sum_{i=1}^k \pi_i r_i^2}=:
\bar r.
\]
Combining the two estimates gives
\[
R_0 \;\ge\; \max\{\tilde{c},\ \bar{r}\}.
\]
Finally, the upper bounds corresponding to the cases where only the column sums or only the row sums are known are sharp, and are attained by the diagonal matrices $D_{\boldsymbol{c}}$ and $D_{\boldsymbol{r}}$, respectively.
\end{proof}

\subsection{{Bounds for $\{\tau_i\}$ and $\bar\tau$}}

\begin{lemma}\label{lemma:fixedpoint}
For $\alpha>0$ and $\beta\geq 0$ let $\phi_{\alpha,\beta}(t):=1-\rm{e}^{-\alpha t-\beta}$ on $[0,1]$, and let $t_{\alpha,\beta}\in[0,1]$ be the maximal solution of
$t=\phi_{\alpha,\beta}(t)$.
Then, if $\beta>0$ or $\alpha >1$ there exists an unique positive solution $t_{\alpha,\beta}\in(0,1]$. Moreover, for every $w\in[0,1]$,
\[
w\ \ge\ \phi_{\alpha,\beta}(w)\ \text{if and only if }\ w\ \ge\ t_{\alpha,\beta},
\quad\text{and}\quad
w\ \le\ \phi_{\alpha,\beta}(w)\ \text{if and only if }\ w\ \le\ t_{\alpha,\beta}.
\]
\end{lemma}
\begin{proof}
Consider $g_{\alpha,\beta}(t):=\phi_{\alpha,\beta}(t)-t
=1-\rm{e}^{-\alpha t-\beta}-t$ on $[0,1]$.
Since $g''_{\alpha,\beta}(t)=-\alpha^2 \rm{e}^{-\alpha t-\beta}<0$,
$g_{\alpha,\beta}$ is strictly concave.
Moreover,
\[
g_{\alpha,\beta}(0)=1-\rm{e}^{-\beta}\ge0,\qquad
g_{\alpha,\beta}(1)=-\rm{e}^{-\alpha-\beta}<0,\qquad
g'_{\alpha,\beta}(0)=\alpha \rm{e}^{-\beta}-1.
\]
\noindent
If $\beta>0$ or $\alpha>1$, we have $g_{\alpha,\beta}(0)>0$ or $g'_{\alpha,\beta}(0)>0$.
By continuity and strict concavity, $g_{\alpha,\beta}$ has a unique zero
$t_{\alpha,\beta}\in(0,1)$ and
\[
g_{\alpha,\beta}(t)
\begin{cases}
> 0, & t\in[0,t_{\alpha,\beta}),\\[2pt]
< 0, & t\in(t_{\alpha,\beta},1].
\end{cases}
\]
\noindent
If $\beta=0$ and $\alpha\le1$, then
$g_{\alpha,\beta}(0)=0$ and $g'_{\alpha,\beta}(0)=\alpha-1\le0$.
By concavity, $g_{\alpha,\beta}(t)\le0$ for all $t\in[0,1]$, hence
$t_{\alpha,\beta}=0$.
\noindent
Finally, since $w\le\phi_{\alpha,\beta}(w)$ is equivalent to $g_{\alpha,\beta}(w)\ge0$,
we obtain $w\le t_{\alpha,\beta}$, and similarly
$w\ge\phi_{\alpha,\beta}(w)$ if and only if $w\ge t_{\alpha,\beta}$.
\end{proof}
For the following results we fix $k\ge 1$ and, for each $i=1,\dots,k$, we define
\[
y_i := \pi_i\,\tau_i\in[0,\pi_i],
\qquad
\boldsymbol{y}:=D_{\boldsymbol{\pi}} \boldsymbol{\tau}.
\]
When a row-stochastic matrix $Q=\{q_{ij}\}\in\mathbb{R}^{k\times k}$ and the row-sum vector
$\boldsymbol{r}=(r_1,\dots,r_k)^\top$ are given, we define
\[
u_i := \sum_{j=1}^k q_{ji}\, r_j\, \pi_j\, \tau_j,
\qquad
\boldsymbol{u}:=Q^{\!\top} D_{\boldsymbol{r}} D_{\boldsymbol{\pi}} \boldsymbol{\tau}.
\]
Note that \(u_i\) is the incoming infectious pressure to type \(i\). We will use the scalar function
\[
\phi_\alpha(t):=1-\rm{e}^{-\alpha t},\qquad \alpha>0,\ t\in[0,1],
\]
and, when the column sums $\boldsymbol{c}=(c_1,\dots,c_k)^\top$ of the
next-generation matrix are fixed, we set for each $i$
\[
\psi_i(s):=\pi_i\,\phi_{c_i/\pi_i}(s)
          =\pi_i\!\left[1-\exp\!\Bigl(-\frac{c_i}{\pi_i}s\Bigr)\right],
\]
and
\[
\zeta_i(s):=\pi_i\,\phi_{1/\pi_i}(s)
           =\pi_i\!\left[1-\exp\!\Bigl(-\frac{s}{\pi_i}\Bigr)\right].
\]
Note that $\phi_{c_i/\pi_i}(s)=\phi_{1/\pi_i}(c_i\,s)=\phi_{c_i}(s/\pi_i)$ therefore $\psi_i(s/c_i)=\zeta_i(s)$ or, equivalently, $\psi_i(s)=\zeta_i(c_i s)$ are both functions that are increasing and strictly concave. Moreover, whenever $(A_Q\boldsymbol{\tau})_i=u_i/\pi_i$, we have
\[
y_i=\pi_i\tau_i=\zeta_i(u_i).
\]

\subsubsection*{Proof of Theorem \ref{thm:fs-col}}
\begin{proof}
Recall $M=P D_{\boldsymbol{c}}$ with $P$ column-stochastic and
\[
A_P=D_{\boldsymbol{\pi}}^{-1}D_{\boldsymbol{c}}P^{\!\top}D_{\boldsymbol{\pi}},\qquad
\boldsymbol{\tau}=\one-\exp(-A_P\boldsymbol{\tau}).
\]
By the above notation we have $\boldsymbol{y}=D_{\boldsymbol{\pi}}\boldsymbol{\tau}$ and
\[
(A_P\boldsymbol{\tau})_i
=(D_{\boldsymbol{\pi}}^{-1}D_{\boldsymbol{c}}P^{\!\top}D_{\boldsymbol{\pi}}\boldsymbol{\tau})_i
=\frac{c_i}{\pi_i}\,(P^{\!\top}\boldsymbol{y})_i.
\]
Hence, for each $i$,
\[
y_i
=\pi_i\!\left[1-\exp\!\Bigl(-\frac{c_i}{\pi_i}\,(P^{\!\top}\boldsymbol{y})_i\Bigr)\right]
=\psi_i\!\big((P^{\!\top}\boldsymbol{y})_i\big),
\]
where $\psi_i$ is increasing by definition. Let
\[
\underline y:=\min_i y_i,
\qquad
\overline y:=\max_i y_i.
\]
Because $P$ is column-stochastic, each $(P^{\!\top}\boldsymbol{y})_i$ is a convex
combination of $\{y_j\}$, hence
\[
\underline y \le (P^{\!\top}\boldsymbol{y})_i \le \overline y
\qquad \text{for all } i=1,\dots,k.
\]
Applying the increasing map $\psi_i$ and using
$y_i=\psi_i((P^{\!\top}\boldsymbol{y})_i)$ yields
\[
\psi_i(\underline y)\le y_i \le \psi_i(\overline y)
\qquad \text{for all } i=1,\dots,k.
\]
Taking the minimum and the maximum in $i$ gives
\[
\underline y \ge \min_i \psi_i(\underline y),\qquad
\overline y \le \max_i \psi_i(\overline y).
\]
Let $\ell\in\displaystyle\arg\min_i\psi_i(\underline y)$ and
$h\in\displaystyle\arg\max_i\psi_i(\overline y)$. Using the relation
$\psi_i(s)=\pi_i\,\phi_{c_i}(s/\pi_i)$ we obtain
\[
\underline y \ge \pi_\ell\,\phi_{c_\ell}(\underline y/\pi_\ell),
\qquad
\overline y \le \pi_h\,\phi_{c_h}(\overline y/\pi_h),
\]
and therefore, by Lemma \ref{lemma:fixedpoint},
\[
\underline y \ge \pi_\ell t_{c_\ell}\ \ge\ y_*,
\qquad
\overline y \le \pi_h t_{c_h}\ \le\ y^*,
\]
where
\[
y_*:=\min_i\{\pi_i t_{c_i}\},
\qquad
y^*:=\max_i\{\pi_i t_{c_i}\}.
\]
Consequently, for every $i$,
\[
y_* \le \underline y \le (P^{\!\top}\boldsymbol{y})_i
          \le \overline {y} \le y^*.
\]
By monotonicity of $\phi_{\alpha}$ for each positive value $\alpha$,
\[
1-\exp\!\Bigl(-\frac{c_i}{\pi_i}y_*\Bigr)
\ \le\
1-\exp\!\Bigl(-\frac{c_i}{\pi_i}(P^{\!\top}\boldsymbol{y})_i\Bigr)
\ \le\
1-\exp\!\Bigl(-\frac{c_i}{\pi_i}y^*\Bigr),
\]
and since $\tau_i=1-\exp\!\left(-\displaystyle\frac{c_i}{\pi_i}(P^{\!\top}\boldsymbol{y})_i\right)$ we obtain
\[
1-\exp\!\left(\displaystyle\frac{c_i}{\pi_i}y_*\right)
\ \le\ \tau_i\ \le\
1-\exp\!\left(-\displaystyle\frac{c_i}{\pi_i}y^*\right)
\qquad \text{for all } i=1,\dots,k.
\]
Let $a\in\displaystyle\arg\min_i\{\pi_i t_{c_i}\}$ and $b\in\displaystyle\arg\max_i\{\pi_i t_{c_i}\}$ and define
\[
P^{\min}_{ij}:=\delta_{ia},\qquad
P^{\max}_{ij}:=\delta_{ib}.
\]
Where $\delta_{ij}$ is the Kronecker delta function.
For $P^{\max}$ we have $(P^{\max})^{\!\top}\boldsymbol{y}=y_b\one$, so
\[
y_i=\pi_i\!\left[1-\exp\!\Bigl(-\frac{c_i}{\pi_i}y_b\Bigr)\right].
\]
For $i=b$ this implies $\tau_b=1-\exp(-c_b\tau_b)$, hence $\tau_b=t_{c_b}$ and
$y_b=\pi_b t_{c_b}=y^*$. Thus with $P=P^{\max}$ we have
\[
\tau_i=1-\exp\!\Bigl(-\frac{c_i}{\pi_i}y^*\Bigr),
\qquad i=1,\dots,k.
\]
Analogously, imposing $P=P^{\min}$ yields
\[
\tau_i=1-\exp\!\Bigl(-\frac{c_i}{\pi_i}y_*\Bigr),
\qquad i=1,\dots,k.
\]
Therefore, the bounds are attained and hence sharp. Multiplying by $\pi_i$ and summing over $i$ gives
\[
\sum_{i=1}^k \pi_i - \sum_{i=1}^k \pi_i \exp\!\Bigl(-\frac{c_i}{\pi_i}\,y_*\Bigr)
\;\le\;
\sum_{i=1}^k \pi_i \tau_i
\;\le\;
\sum_{i=1}^k \pi_i - \sum_{i=1}^k \pi_i \exp\!\Bigl(-\frac{c_i}{\pi_i}\,y^*\Bigr),
\]
and since $\displaystyle\sum_{i=1}^k \pi_i=1$ this yields the stated bounds for
$\bar{\tau}=\displaystyle\sum_{i=1}^k \pi_i \tau_i$.

For sharpness, for all $i$ we have
\[
\tau_i=
\begin{cases}
1-\exp\!\Bigl(-\dfrac{c_i}{\pi_i}\,y_*\Bigr),
  & \text{if } P=P^{\min},\\[4pt]
1-\exp\!\Bigl(-\dfrac{c_i}{\pi_i}\,y^*\Bigr),
  & \text{if } P=P^{\max}.
\end{cases}
\]
Again, multiplying by $\pi_i$ and summing over $i$ produces the equalities for
$\bar{\tau}$, proving that both bounds are attained.
\end{proof}

\subsubsection*{Proof of Theorem \ref{thm:fs-row}}
\begin{proof}
Since $Q$ is row-stochastic and $M=D_{\boldsymbol{r}} Q$,
\[
A_Q \;=\; D_{\boldsymbol{\pi}}^{-1} Q^{\!\top} D_{\boldsymbol{r}} D_{\boldsymbol{\pi}},
\qquad
(A_Q\boldsymbol{\tau})_i \;=\; \frac{1}{\pi_i}\sum_{j=1}^k q_{ji}\,r_j\,\pi_j\,\tau_j .
\]
Hence, the final-size equations read
\begin{equation*}
\tau_i \;=\; 1-\exp\left( -q_{ii}r_i\tau_i-\displaystyle\frac{1}{\pi_i}\displaystyle\sum_{j\neq i}q_{ji}r_j \pi_j \tau_j\right),\qquad i=1,\dots,k.
\end{equation*}
Since $q_{ii}\leq 1$, we have
\begin{equation*}
\tau_i \;\leq\; 1-\exp\left( -r_i\tau_i-\displaystyle\frac{1}{\pi_i}\displaystyle\sum_{j\neq i}q_{ji}r_j \pi_j \tau_j\right),\qquad i=1,\dots,k,
\end{equation*}
where the right-hand side provides an upper bound for $\tau_i$, and this bound is attained if and only if $\tau_i$ is the solution of 
\begin{equation*}
\tau_i \;=\;1-\exp\left( -r_i\tau_i-\displaystyle\frac{1}{\pi_i}\displaystyle\sum_{j\neq i}q_{ji}r_j \pi_j \tau_j\right),\qquad i=1,\dots,k.
\end{equation*}
Equivalently, we can say that to each row-stochastic matrix $Q$ is associated a $\tau_i$, which is bounded above by $\tau^{(i)}_i$, associated with a row-stochastic matrix $Q^{(i)}$ where
\begin{equation*}
q_{ii}^{(i)}=1,\quad q_{ij}^{(i)}=0,\quad 
\text{for all}\,j\neq i,
\end{equation*}
and it is a solution of
\begin{equation*}
\tau_i^{(i)} \;=\; 1-\exp\left( -r_i\tau_i^{(i)}-\displaystyle\frac{1}{\pi_i}\displaystyle\sum_{j\neq i}q_{ji}^{(i)}r_j \pi_j \tau_j^{(i)}\right),\qquad i=1,\dots,k.
\end{equation*}
For simplicity we choose 
\begin{equation*}
q_{ji}^{(i)}= 1-q_{jj},\quad
q_{jj}^{(i)}=q_{jj},\quad
q_{jh}^{(i)}=0\quad
\text{for all}\quad j\neq i,\quad h\,\not\in \{i,\,j \}.
\end{equation*}
With this choice, for $j\neq i$,
\begin{equation*}
\tau_j^{(i)} \;=\; 1-\exp\left\{ -q_{jj}r_j \tau_j^{(i)}\right\},
\end{equation*}
so $\tau_j^{(i)}=t_{q_{jj}r_j}$ and 
\begin{equation}
\tau_i^{(i)} \;=\; 1-\exp\left( -r_i\tau_i^{(i)}-\displaystyle\frac{1}{\pi_i}\displaystyle\sum_{j\neq i}(1-q_{jj})r_j \pi_j t_{q_{jj}r_j}\right).
\end{equation}
We now have to choose $q_{jj}$ in order to maximize the function
\begin{equation*}
f(q_{jj})=(1-q_{jj})\,t_{q_{jj}r_j},
\end{equation*}
where $t_{q_{jj}r_j}$ solves
\begin{equation}
t_{q_{jj}r_j}=1-\exp({-t_{q_{jj}r_j}q_{jj}r_j})
\label{eq:xxx}
\end{equation}
By Lemma \ref{lem:opt-qjj}, if $q_{jj}r_j\leq 1$, we can choose $q_{jj}=1$ without loss of generality. If instead $q_{jj}r_j> 1$, we have $q_{jj}=\displaystyle\frac{\log(r_j)}{r_j-1}.
$
Note that with this choice $q_{jj} r_j>1$ if and only if $r_j>1$ and therefore 
\begin{align*}
\tau_i^{(i)} &
= 1-\exp\left( -r_i\tau_i^{(i)}-\displaystyle\frac{1}{\pi_i}\displaystyle\sum_{j\neq i}(r_j-1-\log(r_j)) \pi_j \delta_{r_j>1}\right)\\
&=1-\exp\left( -r_i\tau_i^{(i)}-\displaystyle\frac{1}{\pi_i}\displaystyle\sum_{j\neq i}\kappa(r_j) \pi_j \right)   \\
&=1-\exp\left( -r_i\tau_i^{(i)}-K_i \right)=t_{r_i, K_i}
\end{align*}
so $\tau_i^{(i)}= t_{r_i, K_i}$, where $\delta_{r_j>1}=1$ if $r_j>1$ and $\delta_{r_j>1}=0$ otherwise. Thus, for Lemma \ref{lemma:fixedpoint}
\begin{equation}
\tau_i\leq \tau_i^{(i)}=t_{r_i, K_i}
\end{equation}
and the bound is attained with the matrix $Q^{(i)}$ defined above.
\end{proof}

\subsubsection*{Proof of Theorem \ref{thm:tot-fs-row}.}
\label{subsubsectio.proof.thm:tot-fs-row}
\begin{proof}
Since \(A_Q = D_{\boldsymbol{\pi}}^{-1} Q^{\!\top} D_{\boldsymbol{r}} D_{\boldsymbol{\pi}}\), we have
\[
(A_Q \boldsymbol{\tau})_i \;=\; \frac{u_i}{\pi_i},
\]
and hence
\[
\tau_i \;=\; 1 - \exp\!\Bigl(-\frac{u_i}{\pi_i}\Bigr), \qquad
y_i=\zeta_i(u_i), \qquad
\bar{\tau} \;=\; \sum_{i=1}^k \zeta_i(u_i).
\]
Since $Q$ is a row-stochastic matrix,
\begin{equation}\label{eq:mass-balance}
U \;=\sum_{i=1}^k u_i
=\sum_{j=1}^k r_j\, \pi_j\, \tau_j
=\sum_{j=1}^k r_j\, y_j
=\sum_{j=1}^k r_j\, \zeta_j(u_j).
\end{equation}
The last identity means that the total produced pressure equals the total received
pressure.

\medskip\noindent
Define
\[
\begin{aligned}
&\mathcal U^+
:=\Bigl\{\,\boldsymbol{u}\ge0:\ \exists\,U\ge0 \text{ with }\displaystyle\sum_{i=1}^{k} u_i=U \text{ and }U=\displaystyle\sum_{i=1}^{k} r_i\,\zeta_i(u_i)\Bigr\},\\
&\mathcal U
:=\Bigl\{\,\boldsymbol{u}\ge0:\ \exists\,Q\text{ row-stochastic},\ \boldsymbol{\tau}=\mathbf 1-\exp\left(-A_Q\boldsymbol{\tau}\right),\
u_i=\displaystyle\sum_{j=1}^k q_{ji} r_j\pi_j\tau_j,\ y_i=\zeta_i(u_i)\Bigr\},\\
&\mathcal U^-
:=\Bigl\{\,\boldsymbol{u}\in\mathcal U:\ Q^{(w)}=\mathbf 1\,\boldsymbol{e}_w^{\!\top}\ \text{    for some   }w,\quad 1\leq w\leq k\Bigr\}.
\end{aligned}
\]
Clearly \(\mathcal U^-\subseteq\mathcal U\subseteq\mathcal U^+\), hence
\begin{equation}\label{eq:chain-U-sets}
\min_{u\in\mathcal U^-}\sum_{i=1}^k \zeta_i(u_i)
\ \ge\
\min_{u\in\mathcal U}\sum_{i=1}^k \zeta_i(u_i)
\ \ge\
\min_{u\in\mathcal U^+}\sum_{i=1}^k \zeta_i(u_i).
\end{equation}
Fix \(U\ge0\). On the simplex \(\{u\ge0:\sum_i u_i=U\}\), Lemma \ref{lemma:bounds} gives
\begin{equation}\label{eq:zeta-simplex-min}
\sum_{i=1}^k \zeta_i(u_i)\ \ge\ \min_{1\le i\le k}\zeta_i(U),
\quad\text{with equality if and only if  }u=U\boldsymbol{e}_w\text{ for some }w.
\end{equation}
Thus any minimizer in \(\mathcal U^+\) must be of the form \(u=U\boldsymbol{e}_w\). Imposing the
closure condition \eqref{eq:mass-balance},
\[
U\;=\; r_w\,\zeta_w(U)\;=\;r_w\,\pi_w\left[1-\exp\left(-\displaystyle\frac{U}{\pi_w}\right)\right].
\]
Let \(x:=U/\pi_w\). Then \(x=r_w(1-e^{-x})\), whose maximal solution in \([0,1]\) is
\(x=t_{r_w}\). Hence
\[
U=\pi_w r_w t_{r_w}
\quad\text{and}\quad
\zeta_w(U)=\pi_w t_{r_w}.
\]
Now, minimizing over \(w\) yields
\begin{equation}\label{eq:Uplus-min}
\min_{u\in\mathcal U^+}\sum_{i=1}^k \zeta_i(u_i)
= \min_{1\le w\le k}\ \pi_w\, t_{r_w}.
\end{equation}
For \(Q^{(w)}=\mathbf 1\,\boldsymbol{e}_w^{\!\top}\) (row-stochastic), each row equals \(\boldsymbol{e}_w^{\!\top}\), hence
\[
u_i=\sum_{j=1}^k Q^{(w)}_{ji} r_j \pi_j \tau_j
=\delta_{iw}\sum_{j=1}^k r_j \pi_j \tau_j
=\delta_{iw}\,U,
\]
so \(\boldsymbol{u}=U\boldsymbol{e}_w\). The same closure relation gives \(U=\pi_w r_w t_{r_w}\) and
\[
\sum_{i=1}^k \zeta_i(u_i)=\zeta_w(U)=\pi_w t_{r_w}.
\]
Therefore
\begin{equation}\label{eq:Uminus-min}
\min_{u\in\mathcal U^-}\sum_{i=1}^k \zeta_i(u_i)
= \min_{1\le w\le k}\ \pi_w\, t_{r_w}.
\end{equation}
From \eqref{eq:chain-U-sets}, \eqref{eq:Uplus-min} and \eqref{eq:Uminus-min} we obtain
\[
\min_{u\in\mathcal U}\sum_{i=1}^k \zeta_i(u_i)
\;=\;
\min_{1\le i\le k}\ \pi_i\, t_{r_i}.
\]
Since for any admissible \(Q\) we have $\bar{\tau}=\displaystyle\sum_{i=1}^k \zeta_i(u_i)$, it follows that
\[
 \min_{1\le i\le k}\ \pi_i\, t_{r_i}\;\leq\;   \bar{\tau},
\]
for all row-stochastic \(Q\), with equality achieved by
$Q^{(w)}=\mathbf 1 \boldsymbol{e}_w^{\!\top}$ for the minimizing index \(w\).

\medskip
For the upper bound we note that every row-stochastic $Q$ generates some vector
$\boldsymbol{u}\in\mathcal U$, and conversely any $\boldsymbol{u}\in\mathcal U$
is induced by (at least) one row-stochastic $Q$: if $\sum_i u_i>0$, set
$\boldsymbol{q}:=\boldsymbol{u}/\sum_i u_i=\boldsymbol{u}/ U$ and let $Q=\mathbf 1\,\boldsymbol{q}^{\!\top}$;
if $\sum_i u_i=0$, any row-stochastic $Q$ suffices. Therefore
\begin{equation}\label{eq:Q-sup}
\sup_{Q\ \text{row-stochastic}}\bar{\tau}
\;=\;
\max_{\boldsymbol{u}\in\mathcal U}\ \sum_{i=1}^k \zeta_i(u_i).
\end{equation}
By the definition of $\mathcal U$, this is equivalent to the constrained optimization
problem
\[
\max_{\boldsymbol{u}\ge 0}\ \Psi(\boldsymbol{u})
\quad\text{s.t.}\quad h(\boldsymbol{u})=0,
\]
where
\[
\Psi(\boldsymbol{u})=\sum_{i=1}^{k}\zeta_i(u_i),
\qquad
h(\boldsymbol{u})=\sum_{i=1}^{k}u_i-\sum_{i=1}^{k}r_i\,\zeta_i(u_i).
\]
Therefore, by \ref{lem:kkt}, if \(r_i\le 1\) for all \(i\), the unique solution is
$\boldsymbol{u}=\boldsymbol{0}$ and hence $\bar{\tau}=0$ for any choice of $Q$.
If there exists \(i\) with \(r_i>1\), there exists a unique $\lambda^\star>0$ and the
maximizer is given componentwise by
\[
u_i^\star=\pi_i\left[\ \log\left(\displaystyle\frac{1+\lambda^\star r_i}{\lambda^\star} \right)\ \right]_+,
\qquad i=1,\dots,k,
\]
where $\lambda^\star$ is the unique solution of
\[
\sum_{i=1}^k\!\left[\ \pi_i\,\log\!\left(r_i+\frac{1}{\lambda}\right)
\;-\;r_i\,\pi_i\Bigl(1-\frac{\lambda}{1+\lambda r_i}\Bigr)\right]_+=0.
\]
For this choice we have, using the definition of $\zeta_i$,
\[
\bar{\tau}
\le
\sum_{i=1}^k
\zeta_i(u_i^\star)
=
\sum_{i=1}^k\pi_i\!\left[1-\exp\!\Bigl(-\frac{u_i^\star}{\pi_i}\Bigr)\right]
=
\sum_{i=1}^k \pi_i\Bigl[\ 1-\frac{\lambda^\star}{1+\lambda^\star r_i}\ \Bigr]_+.
\]
To show that this upper bound is sharp it remains to construct a row-stochastic
$Q^\star$ such that the corresponding vector of pressures is exactly
$\boldsymbol{u}^\star$. Let
\[
U^\star=\sum_{i=1}^k u_i^\star
\qquad\text{and}\qquad
\boldsymbol{q}^\star=\boldsymbol{u}^\star/U^\star.
\]
Define $Q^\star=\mathbf 1\,(\boldsymbol{q}^\star)^{\!\top}$, which is row-stochastic.
Set $y_i^\star:=\zeta_i(u_i^\star)$ and $s_i^\star:=r_i\,y_i^\star$ for $i=1,\dots,k$,
so that
\[
\sum_{i=1}^k s_i^\star
=\sum_{i=1}^k r_i\,\zeta_i(u_i^\star)
=\sum_{i=1}^k u_i^\star
=U^\star,
\]
where we used the constraint $h(\boldsymbol{u}^\star)=0$.
Then
\[
Q^{\star\top}\boldsymbol{s}^\star
=(\boldsymbol{q}^\star\,\mathbf 1^\top)\boldsymbol{s}^\star
= \boldsymbol{q}^\star\,(\mathbf 1^\top \boldsymbol{s}^\star)
= \boldsymbol{q}^\star\,U^\star
= \boldsymbol{u}^\star,
\]
so $\boldsymbol{u}^\star$ is realized by $Q^\star$.
Finally, define $\tau_i^\star:=1-\exp(-u_i^\star/\pi_i)$, so that
$y_i^\star=\pi_i\tau_i^\star$ and
\[
(A_{Q^\star}\boldsymbol{\tau}^\star)_i
=\frac{1}{\pi_i}\sum_{j=1}^k q^{\star}_{ji}r_j\pi_j\tau_j^\star
=\frac{u_i^\star}{\pi_i},
\]
showing that $\boldsymbol{\tau}^\star$ solves the final-size equation
$\boldsymbol{\tau}=\mathbf 1-\exp\left(-A_{Q^\star}\boldsymbol{\tau}\right)$.
Therefore the corresponding total final size is
\[
\bar{\tau}
=\sum_{i=1}^k \zeta_i(u_i^\star)
=\sum_{i=1}^k \pi_i\Bigl[\ 1-\frac{\lambda^\star}{1+\lambda^\star r_i}\ \Bigr]_+,
\]
and the upper bound is attained.
\end{proof}

\subsubsection{Proof of Theorem \ref{thm:DBk=2} }
\label{subsec:thm:DBk=2:proof}
\begin{proof}
(a) Recall that $\varphi =\frac{\pi_1}{\pi_2}$.  The eigenvalues of $M$ satisfy
\[
(r_1-\theta-\lambda)(r_2-\varphi\theta-\lambda)-\varphi\theta^2=0,
\]
which, on rearranging yields,
\[
\lambda^2-[r_1+r_2-(\varphi+1)\theta]\lambda+r_1r_2-(\varphi r_1+r_2)\theta=0.
\]
Hence, $R_0=R_0(\theta)$, where
\begin{equation}
\label{equ:DBk=2R0}
R_0(\theta)=\frac{1}{2}\left[r_1+r_2-(\varphi+1)\theta+\sqrt{[r_1+r_2-(\varphi+1)\theta]^2-4r_1r_2+4(\varphi r_1+r_2)\theta}\right].
\end{equation}
Differentiating,
\begin{align*}
R_0^{\prime}(\theta)&=\frac{1}{2}\left[-(\varphi+1)+\frac{(\varphi+1)^2\theta-(\varphi+1)(r_1+r_2)+2(\varphi r_1+r_2)}{\sqrt{[r_1+r_2-(\varphi+1)\theta]^2-4r_1r_2+4(\varphi r_1+r_2)\theta}}\right]\\
&=\frac{1}{2}\left[-(\varphi+1)+\frac{(\varphi+1)^2\theta-(\varphi-1)(r_2-r_1)}{\sqrt{[r_1+r_2-(\varphi+1)\theta]^2-4r_1r_2+4(\varphi r_1+r_2)\theta}}\right],
\end{align*}
so $R_0^{\prime}(\theta)<0$ if and only if
\begin{equation}
\label{equ:R_0<0}    
(\varphi+1)\theta -\left(\frac{\varphi-1}{\varphi+1}\right)(r_2-r_1)<\sqrt{[r_1+r_2-(\varphi+1)\theta]^2-4r_1r_2+4(\varphi r_1+r_2)\theta}.
\end{equation}
The inequality in \eqref{equ:R_0<0} clearly holds if $(\varphi+1)\theta <\left(\frac{\varphi-1}{\varphi+1}\right)(r_2-r_1)$, otherwise squaring both sides of (\ref{equ:R_0<0}) and rearranging shows that (\ref{equ:R_0<0})  holds if and only if $\varphi>0$, which is clearly true. 
Hence, $R_0(\theta)$ is strictly decreasing on $[0, \theta_{\max}]$, and sharp lower and upper bounds for $R_0$ are given by $R_0(\theta_{\max})$ and $R_0(0)=r_2$, respectively. The lower bounds in~\eqref{equ:DBk=2R0boundsa} and~\eqref{equ:DBk=2R0boundsb} are obtained by substituting $\theta=\theta_{\max}$ in ~\eqref{equ:DBk=2R0}, noting that $\theta_{\max}=r_1$ if $\varphi r_1 \le r_2$ and $\theta_{\max}=\varphi^{-1}r_2$ if $\varphi r_1 \ge r_2$. 

\medskip
(b) It follows from~\eqref{eq:finalsizeDB} that $(\tau_1(\theta), \tau_2(\theta))$ solves
\begin{equation}
\label{equ:tau1r}
1-\tau_1=\exp\left(-[\tau_1(r_1-\theta)+\tau_2\theta]\right)
\end{equation}
and
\begin{equation}
\label{equ:tau2r}
1-\tau_2=\exp\left(-[\tau_1\varphi\theta+\tau_2(r_2-\varphi\theta)]\right).
\end{equation}
Suppose that there exists $\theta \in(0, \theta_{\max}]$  such that $\tau_1(\theta)=\tau_2(\theta)$ ($=\tau$ say).  Then equating the right-hand sides of (\ref{equ:tau1r}) and \eqref{equ:tau2r} implies $\tau r_1=\tau r_2$, whence either $\tau=0$ or $r_1=r_2$. The latter is impossible, so, noting that $\tau_1(0)<\tau_2(0)$ as $r_1<r_2$ and $r_2>1$, if $R_0(\theta_{\max})>1$ then 
$\tau_1(\theta)<\tau_2(\theta)$ for all $\theta \in [0,\theta_{\max}]$.  Alternatively, if $R_0(\theta_{\max}) \le 1$, then $\tau_1(\theta)<\tau_2(\theta)$ for all $\theta \in [0,\theta_*)$ and $\tau_1(\theta)=\tau_2(\theta)=0$ for all $\theta \in [\theta_*, \theta_{\max}]$.  

\medskip 
Differentiating (\ref{equ:tau1r}) and (\ref{equ:tau2r}) with respect to $\theta$ yields
\begin{equation}
\label{equ:tau1dashr}
\tau_1^{\prime}=(1-\tau_1)[\tau_2-\tau_1+(r_1-\theta)\tau_1^{\prime}+\theta \tau_2^{\prime}]
\end{equation}
and
\begin{equation}
\label{equ:tau2dashr}
\tau_2^{\prime}=(1-\tau_2)[\varphi(\tau_1-\tau_2)+\varphi\theta\tau_1^{\prime}+(r_2-\varphi\theta) \tau_2^{\prime}].
\end{equation}
Dividing \eqref{equ:tau1dashr} by $(1-\tau_1)/\varphi$ and \eqref{equ:tau2dashr} by $1-\tau_2$, then adding yields, after rearrangement, 
\begin{equation}
\label{equ:tauprimebalr}
\varphi\tau_1^{\prime}(\theta) f_1(\theta)+\tau_2^{\prime}(\theta) f_2(\theta)=0,
\end{equation}
where $f_i(\theta)$ is defined at (\ref{equ:fidef}).

\medskip
Note that, if $r_1<1$, then $f_1(\theta)>0$ for all $\theta \in [0, \theta_{\max}]$.  Alternatively, suppose $r_1 \ge 1$.  Then, since $\tau_1(\theta)(r_1-\theta)+\tau_2(\theta)\theta > \tau_1(\theta)r_1$ for all $\theta \in (0, \theta_*)$, it follows from (\ref{equ:tau1r}) that $\tau_1(\theta) > \tau_1(0)$ for all $\theta \in (0, \theta_*)$ (cf.~Lemma~\ref{lemma:fixedpoint}).  Hence, $f_1(\theta)>0$ for all $\theta \in (0, \theta_*)$, since $f_1(\theta)$ is strictly increasing in $\tau_1(\theta)$.

\medskip
Setting $\theta=0$ in (\ref{equ:tau1dashr}) and \eqref{equ:tau2dashr} and rearranging
yields
\begin{equation}
\label{equ:tau1'0r}
\tau_1^{\prime}(0)=\frac{\tau_2(0)-\tau_1(0)}{f_1(0)}>0
\end{equation}
and
\begin{equation}
\label{equ:tau2'0r}
\tau_2^{\prime}(0)=\varphi\left(\frac{\tau_1(0)-\tau_2(0)}{f_2(0)}\right)<0.
\end{equation}
Observe from (\ref{equ:tau1dashr}) that a necessary condition for $\tau_1^{\prime}(\theta)=\tau_2^{\prime}(\theta)=0$ is $\tau_1(\theta)=\tau_2(\theta)$, which is impossible for $\theta \in [0, \theta_*)$. Thus, since $f_1(\theta)>0$ for all $\theta \in (0, \theta_*)$, it follows from (\ref{equ:tauprimebalr}) and (\ref{equ:tau2'0r}) that $\tau_2^{\prime}(\theta)<0$ for all $\theta \in (0, \theta_*)$, proving part (b)(i).

\medskip
Recalling that $f_1(\theta)>0$ and $\tau_2^{\prime}(\theta)<0$ for all $\theta \in (0, \theta_*)$, it follows from \eqref{equ:tauprimebalr} that, for such $\theta$,
\begin{equation}
\label{equ:tau1'=0r}
\tau_1^{\prime}(\theta)=0 \text{  if and only if  }f_2(\theta)=0.
\end{equation}
Hence, if $f_2(\theta_*) \ge 0$ then, in view of \eqref{equ:tau1'0r}, $\tau_1^{\prime}(\theta)>0$ for all $\theta \in [0, \theta_*)$.  Alternatively, if $f_2(\theta_*) < 0$, there exists a unique $\theta_0 \in (0, \theta_*)$ such that $f_2(\theta_0)=0$.  It follows that $\tau_1(\theta)$ has a unique turning point in $(0, \theta_*)$ at $\theta_0$, which in view of (\ref{equ:tau1'0r}) and (\ref{equ:tauprimebalr}) must be a maximum, proving part (b)(ii).  

\medskip
(c)
We prove first the assertions concerning the functions $g(r)$ and $\hat{r}_2(r)$ made just before the statement of Theorem~\ref{thm:DBk=2}.
For $r>1$, 
\begin{align*}
 g(r)>0 \qquad&\text{if and only if}\qquad t_r>1-\frac{1}{r}\\
 &\text{if and only if}\qquad 1-\frac{1}{r}<1-{\rm e}^{-r\left(1-\frac{1}{r}\right)}\qquad(\text{using Lemma~\ref{lemma:fixedpoint}})\\
 &\text{if and only if}\qquad {\rm e}^{r-1}>r,
\end{align*}
which clearly holds, since ${\rm e}^{r-1}>1+(r-1)$ as $r>1$. Thus, $g(r)>0$ for $r>1$.

\medskip
Also, for $r>1$, writing $t^{\prime}_r$ for the derivative of $t_r$ with respect to $r$,
\[
g^{\prime}(r)=\frac{t^{\prime}_r}{(1-t_r)^2}-1.
\]
Differentiating both sides of $1-t_r={\rm e}^{-rt_r}$ with respect to $r$ yields, after a little algebra, 
\[
t^{\prime}_r=\frac{t_r(1-t_r)}{1-r(1-t_r)}.
\]
(Note $g(r)>0$ implies $1-r(1-t_r)>0$.)  Thus,
\[
g^{\prime}(r)=\frac{t_r}{(1-t_r)(1-r(1-t_r))}-1
\]
and, letting $z(=z(r))=1-t_r$ $(\in (0,1))$,
\[
g^{\prime}(r)>0 \qquad\text{if and only if}\qquad 1-z >z(1-rz) \qquad\text{if and only if}\qquad 1-2z+rz^2>0,
\]
which clearly holds as $r>1$.  Thus, $g$ is strictly increasing on $[1, \infty)$. 

\medskip
Turning to $\hat{r}_2(r)$, note that, since $g$ is strictly increasing on $[1, \infty)$, $\hat{r}_2(r)$ is well-defined for $r \in (0,1)$ if there exists $r>1$ such that  $g(r) \ge 1$. Now $g(r) \ge 1 \,\text{if and only if}\, t_r \ge r/(1+r)$.  By Lemma~\ref{lemma:fixedpoint}, 
\[
t_r
\ge \frac{r}{1+r} \qquad\text{if and only if}\qquad {\rm e}^{-\frac{r^2}{1+r}} \le \frac{1}{1+r} \qquad\text{if and only if}\qquad {\rm e}^{\frac{r^2}{1+r}} \ge 1+r.
\]
Now
\[ 
{\rm e}^{\frac{r^2}{1+r}}\ge 1+\frac{r^4}{2(1+r)^2}>1+\frac{r^4}{8r^2}\ge 1+r \quad\text{for } r \ge 8,
\]
as required.

\medskip
Suppose that $\theta \in (0, \theta_*)$.  Then $f_1(\theta) > 0$ and (\ref{equ:tauprimebalr}) implies $\tau_1^{\prime}(\theta)=-\frac{1}{\varphi}(f_2(\theta)/f_1(\theta))\tau_2^{\prime}(\theta)$, whence
\[
\bar{\tau}^{\prime}(\theta)=\pi_1\tau_1^{\prime}(\theta)+\pi_2\tau_2^{\prime}(\theta)
=-\frac{\pi_1}{\varphi}\frac{f_2(\theta)}{f_1(\theta)}\tau_2^{\prime}(\theta)+\pi_2 \tau_2^{\prime}(\theta)
=\pi_2 \left(1-\frac{f_2(\theta)}{f_1(\theta)}\right)\tau_2^{\prime}(\theta).
\]
Now $\tau_2^{\prime}(\theta)<0$, so
\begin{equation}
\label{equ:signtau'r}
\bar{\tau}^{\prime}(\theta)
\begin{cases}
>0& \text{ if   } f_1(\theta)<f_2(\theta),\\
=0& \text{ if   } f_1(\theta)=f_2(\theta),\\
<0 & \text{ if   } f_1(\theta)>f_2(\theta) .
\end{cases}
\end{equation}
Suppose that $r_1 \in (0,1)$ and $r_2 \in (1, \hat{r}_2(r_1)]$.  Then $f_1(0) \ge f_2(0)$.  Note from \eqref{equ:fidef} that $f_i(\theta)$ is increasing with $\tau_i(\theta)$ $(i=1,2)$. Thus, using (\ref{equ:tau1'0r}) and \eqref{equ:tau1'=0r}, as $\theta$ increases from $0$, $f_1(\theta)$ is increasing while $f_2(\theta)>0$.  Moreover, since $f_2(\theta)$ is strictly decreasing in $\theta$ and $f_1(\theta)>0$ for all $\theta \in [0, \theta_*]$, we have that $f_1(\theta)>f_2(\theta)$ for all $\theta \in [0, \theta_*]$ and, using (\ref{equ:signtau'r}), $\bar{\tau}(\theta)$ is strictly decreasing on $[0, \theta_*]$, proving part (c)(i).

\medskip
Suppose that $r_1 \in (0,1)$ and $r_2 >\hat{r}_2(r_1)$, or $r_2 \ge 1$.  Then $f_1(0)< f_2(0)$, so either (A) $f_1(\theta) < f_2(\theta)$ for all $\theta \in (0, \theta_*)$ or (B) there exists $\theta \in (0, \theta_*)$ such that $f_1(\theta)=f_2(\theta)$.  In case (A), \eqref{equ:signtau'r} implies that $\bar{\tau}(\theta)$ is strictly increasing on $[0, \theta_*]$.  Moreover, this case can occur only if $R_0(\theta_{\max})>1$, so $\theta_*=\theta_{\max}$, since $\bar{\tau}(\theta_*)=0$ if $R_0(\theta_{\max})\le 1$.  In case (B), let $\theta_1=\min(\theta \in (0, \theta_*):f_1(\theta)=f_2(\theta))$.  Recall from the proof of part (b) that $\tau^{\prime}_1(\theta)$, and hence also $f^{\prime}_1(\theta)$, is strictly positive for all $\theta \in (0, \theta_*)$ if $f_2(\theta_*)>0$, and also that $f_1(\theta)>0$ for all $\theta \in (0, \theta_*)$.  It follows that $f_1(\theta)>f_2(\theta)$ for all $\theta \in (\theta_1, \theta_*)$ and $\theta_1$ is the only root of $f_1(\theta)-f_2(\theta)$ in $(0, \theta_*)$.  Thus, $f_1(\theta)<f_2(\theta)$ for $\theta \in [0, \theta_1)$ and $f_1(\theta)>f_2(\theta)$ for $\theta \in (\theta_1, \theta_*)$, and it follows from \eqref{equ:signtau'r} that $\bar{\tau}(\theta)$ is strictly increasing on $[0, \theta_1]$ and strictly decreasing on $[\theta_1, \theta_*]$. Moreover, since $\theta_1$ is unique, case (A) occurs if $f_1(\theta_*) \le f_2(\theta_*)$ and case (B) occurs if  $f_1(\theta_*) > f_2(\theta_*)$, thus completing the proof of part (c)(ii).

\medskip
Turning to part (c)(iii), suppose that $f_1(\theta_1)=f_2(\theta_1)$ for some $\theta_1 \in [0, \theta_*]$.  We show that $\bar{\tau}(\theta_1)$ attains the upper bound for $\bar{\tau}$ assuming only that the row sums of $M$ are $r_1$ and $r_2$.  Let
\[
\lambda^*=\frac{1-\tau_1(\theta_1)}{1-r_1(1-\tau_1(\theta_1))}=\frac{1}{f_1(\theta_1)}=\frac{1}{f_2(\theta_1)}=\frac{1-\tau_2(\theta_1)}{1-r_2(1-\tau_2(\theta_1))}.
\]
Note that $\lambda^*\in (0, \infty)$, since $f_1(\theta) \ge 0$ for all $\theta \in [0, \theta_*]$, with strict inequality unless $r_1=1$ and $\theta=0$.  However, if $r_1=1$, then $\theta_1 \ne 0$, since $f_2(\theta_1)=f_1(\theta_1)$ and $f_2(0) >f_1(0)$.
Now, 
\[
y_i(\lambda^*)=\pi_i\left[1-\frac{\lambda^*}{1+\lambda^*r_i}\right]_+=\pi_i\left(1-\frac{\lambda^*}{1+\lambda^*r_i}\right) =\pi_i \tau_i(\theta_1)  \qquad(i=1,2).
\]
Hence,
\[
u_i(\lambda^*)=\pi_i\left[\log\left(\frac{1+\lambda^*r_i}{\lambda^*}\right)\right]_+=-\pi_i\log(1-\tau_i(\theta_1)) \qquad(i=1,2).
\]
Thus,
\begin{align*}
 F(\lambda^*)&=\sum_{i=1}^2 u_i(\lambda^*)-\sum_{i=1}^2 r_i y_i(\lambda^*)\\
 &=-\pi_1\log(1-\tau_1(\theta_1))-\pi_2\log(1-\tau_2(\theta_1))-\pi_1r_1\tau_1(\theta_1)-\pi_2 r_2 \tau_2(\theta_1).
 \end{align*}
Using \eqref{equ:tau1r} and \eqref{equ:tau2r},
\begin{align*}
&-\pi_1\log(1-\tau_1(\theta_1))-\pi_2\log(1-\tau_2(\theta_1))\\
&\qquad=\pi_1\left[\tau_1(\theta_1)(r_1-\theta_1)+\tau_2(\theta_1)\theta_1\right]+\pi_2\left[\tau_1(\theta_1)\varphi\theta_1+\tau_2(\theta_1)(r_2-\varphi\theta_1)\right]\\
&\qquad=\pi_1\tau_1(\theta_1)r_1+\pi_2\tau_2(\theta_1)r_2+\pi_1\theta_1[\tau_2(\theta_1)-\tau_1(\theta_1)]-\varphi\pi_2\theta_1[\tau_2(\theta_1)-\tau_1(\theta_1)]\\
&\qquad=\pi_1\tau_1(\theta_1)r_1+\pi_2\tau_2(\theta_1)r_2,
\end{align*}
since $\pi_1=\varphi\pi_2$.
Therefore, $F(\lambda^*)=0$, so $\bar{\tau}(\theta_1)$ attains the upper bound for $\bar{\tau}$ given in Theorem~\ref{thm:tot-fs-row}.
\end{proof}

\subsubsection{Proof Theorem \ref{thm:DBk>2} }
\label{subsec:thm:DBk>2:proof}
\begin{proof}
Let $\lambda^*=c^{-1}$.  Then, for $i \in A$, 
\[
\frac{1+\lambda^*r_i}{\lambda^*}=\frac{1+f_i^{-1}r_i}{f_i^{-1}}=f_i+r_i=\frac{1}{1-\tau_i}>1
\]
and for $i \notin A$,
\[
\frac{1+\lambda^*r_i}{\lambda^*}=\frac{1+c^{-1}r_i}{c^{-1}}=c+r_i \le 1.
\]
Thus,
\[
\pi_i\left[\log\frac{1+\lambda^*r_i}{\lambda^*}\right]_+=\begin{cases}
    -\pi_i \log(1-\tau_i) &\quad \text{if }  i \in A,\\
   0\;\;(=-\pi_i \log(1-\tau_i))& \quad \text{if } i \notin A , 
\end{cases}
\]
and
\[
\pi_i\left[1-\frac{\lambda^*}{1+\lambda^* r_i}\right]_+=\begin{cases}
    \pi_i \tau_i &\quad \text{if }  i \in A,\\
   0 \;\;(=\pi_i \tau_i)& \quad \text{if } i \notin A . 
\end{cases}
\]
Hence,
\begin{align*}
\sum_{j=1}^k \pi_j\left[\log\frac{1+\lambda^*r_j}{\lambda^*}\right]_+&=-\sum_{j=1}^k \pi_j \log(1-\tau_j)\\
&=\sum_{j=1}^k \sum_{i=1}^k \pi_i \tau_i m_{ij}\qquad(\text{using }\eqref{eq:finalsize})\\
&=\sum_{i=1}^k \pi_i \tau_i r_i\\
&=\sum_{i=1}^k r_i \pi_i\left[1-\frac{\lambda^*}{1+\lambda^* r_i}\right]_+.
\end{align*}
Therefore, $\lambda^*$ satisfies (and hence is the unique solution) of~\eqref{eq:lambdastarequ} and $\tau^*=\sum_{i=1}^k \pi_i \tau_i=\bar{\tau}$. Thus, $\bar{\tau}$ attains the upper bound for $\bar{\tau}$ given in Theorem \ref{thm:tot-fs-row}.
\end{proof}

\section{Discussion}

In this paper we derived bounds on the basic reproduction number $R_0$, the final size vector
$\boldsymbol{\tau}=(\tau_1,\dots,\tau_k)^{\top}$ and the total final size
$\bar{\tau}=\sum_{i=1}^k \pi_i \tau_i$ when the next-generation matrix $M$ is only partly known through its
row or column sums. We consider both the case of having a general next generation matrix $M$ and the special case where $M$ satisfies detailed balance: $\pi_i m_{ij}=\pi_jm_{ji}$ for all $i$ and $j$, which is satisfied when the only difference between types lies in how they mix.

The situation where $M$ is unrestricted turned out to be simpler and sharp bounds were obtained. The upper and lower bounds were however quite far from each other, and further apart the wider the spread of the row (column) sums were. The lower bound was trivial (=0) whenever one type has row (column) sum smaller than 1.

The detailed balance situation was harder and our bounds were in general not sharp. Still the upper and lower bounds were narrower than compared with the general $M$. Some surprising results were also observed for this situation: in a 2-type situation fixing one row sum $r_1$, the lower bound for $R_0$ as well as $\bar\tau$ was in fact \emph{decreasing} in the second row sum $r_2$ (for small $r_2$). It may even happen that the model is supercritical for very small $r_2$ and subcritical if the second type has more contacts.

When restricting to NGMs derived from contact studies, so that $M$ satisfies detailed
balance, the general bounds remain valid but can often be improved.  Detailed balance allows a
symmetrization and hence brings spectral tools for the $R_0$ analysis. For the final size, however, the detailed-balance constraint must interact with the nonlinearity of the
fixed-point equation, making the problem considerably more complicated. In the case
$k=2$ the admissible set of the NGMs is one dimensional and a better analysis is possible; for $k>2$, the geometry is higher dimensional and the conjecture in Section~\ref{subsubsec:DBMk>2} yields further theoretical aspects that could be explored.

The paper admittedly considers a simplified model in which the only heterogeneity is between different types of individual, and where no additional structure or randomness is considered. 

The paper considers a data situation in which the next-generation matrix $M$ is only partly known through its row or column sums. It would be interesting to analyse the situation where both row \emph{and} column sums are observed for each type. Of course the intersection of the rows-only and columns-only bounds could be used, but for most situations these combined bounds will not be sharp. 

In Section \ref{Bounds-empirical}, we considered an empirical  situation where not only row sums are observed, but also some partial information about  row sums is available. It would hence be of interest to derive bounds for this and similar situations. 

A natural extension is when additional data provide
information such as homogeneous or assortative mixing constraints,  or if information beyond one-step transmission, for example, constraints consistent with two generations (i.e.\
with $M^2$), as may arise from contact tracing summaries, the admissible set of NGMs shrinks and the bounds may tighten; in principle the same idea extends to $n$-step constraints.

\newpage
\bibliography{sn-bibliography}

@article{BB25,
  author  = "Britton, T. and Ball, F.",
  title   = "Improving the use of social contact studies in epidemic modelling",
  journal = "Epidemiology",
  volume  = "36",
  pages   = "660--667",
  year    = "2025",
  doi     = "10.1097/EDE.0000000000001876"
}

@article{DHB13,
  author    = "Diekmann, Odo and Heesterbeek, Johan A. P. and Britton, Tom",
  title     = "Mathematical {T}ools for {U}nderstanding {I}nfectious {D}isease {D}ynamics",
  journal = "Princeton University Press",
  year      = "2013",
  doi       = "10.23943/princeton/9780691155395.001.0001"
}

@article{M08,
  author  = "Mossong, J. and Hens, N. and Jit, M. and Beutels, P. and Auranen, K. and others",
  title   = "Social contacts and mixing patterns relevant to the spread of infectious diseases",
  journal = "PLoS Medicine",
  volume  = "5",
  number  = "3",
  pages   = "e74",
  year    = "2008",
  doi     = "10.1371/journal.pmed.0050074"
}

@article{MSW18,
  author  = "Magal, Pierre and Seydi, Ousmane and Webb, Glenn",
  title   = "Final size of a multi-group {S}{I}{R} epidemic model: Irreducible and non-irreducible modes of transmission",
  journal = "Mathematical Biosciences",
  volume  = "301",
  pages   = "59--67",
  year    = "2018",
  doi     = "10.1016/j.mbs.2018.03.020"
}

@misc{SCD,
  author       = "{Social Contact Data Repository}",
  title        = "SocialContactData.org",
  howpublished = "\url{http://www.socialcontactdata.org/data/}",
  note         = "Accessed 7 Nov 2025"
}

@article{CW,
  author  = "Xing, R. and Zhou, B.",
  title   = "Sharp bounds for the spectral radius of nonnegative matrices",
  journal = "Linear Algebra and its Applications",
  volume  = "449",
  pages   = "194--209",
  year    = "2014",
  doi     = "10.1016/j.laa.2014.02.031"
}

@article{RR,
  author    = "Horn, Roger A. and Johnson, Charles R.",
  title     = "Matrix {A}nalysis",
  publisher = "Cambridge University Press",
  year      = "2013",
  doi       = "10.1017/CBO9781139020411"
}

@article{KKT,
  author    = "Boyd, Stephen and Vandenberghe, Lieven",
  title     = "Convex Optimization",
  publisher = "Cambridge University Press",
  year      = "2004",
  doi       = "10.1017/CBO9780511804441"
}

@article{Bauer1958,
  author  = "Bauer, H.",
  title   = "Minimalstellen von Funktionen und Extremalpunkte",
  journal = "Archiv der Mathematik",
  volume  = "9",
  pages   = "389--393",
  year    = "1958",
  doi     = "10.1007/BF01898615"
}

@article{SDP,
    author = "Vandenberghe, Lieven and Boyd, Stephen",
    title = "Semidefinite Programming",
    journal = "SIAM Review",
    volume = "38",
    pages = "49-95",
    year = "1996",
    doi = "10.1137/1038003"
}

@article{WIL,
  author  = {Willem, L. and  Van Kerckhove, K. and  Chao, D.L. and Hens, N. and Beutels, P.},
  title   = {A Nice Day for an Infection? Weather Conditions and Social Contact Patterns Relevant to Influenza Transmission},
  journal = {PLOS ONE},
  year    = {2012},
  volume  = {7},
  number  = {11},
  pages   = {e48695},
  doi     = {10.1371/journal.pone.0048695}
}

@article{VIGGO,
  author  = {Andreasen, V.},
  title   = {The final size of an epidemic and its relation to the basic reproduction number},
  journal = {Bull. Math. Biol.},
  year    = {2011},
  volume  = {73},
  number  = {73},
  pages   = {2305-21},
  doi     = {10.1007/s11538-010-9623-3}
}

@article{CLANCY,
  author  = {Clancy, D. and Pearce, C.J.},
  title   = {The effect of population heterogeneities upon spread of infection.},
  journal = {J. Math. Biol.},
  year    = {2013},
  volume  = {67},
  number  = {},
  pages   = {963–987},
  doi     = {10.1007/s00285-012-0578-x}
}

@article{MA,
  author  = {Ma, J. and Earn, D.J.D.},
  title   = {Generality of the Final Size Formula for an Epidemic of a Newly Invading Infectious Disease.},
  journal = {Bull. Math. Biol.},
  year    = {2006},
  volume  = {68},
  number  = {},
  pages   = {679–702},
  doi     = {10.1007/s11538-005-9047-7}
}

@article{Boyd2,
  author  = {Boyd, Stephen and Diaconis, Persi and Xiao, Lin},
  title   = {Fastest Mixing Markov Chain on a Graph.},
  journal = {SIAM Review},
  year    = {2004},
  volume  = {46},
  number  = {},
  pages   = {667–689},
  doi     = {10.1137/s0036144503423264}
}

@article{S-T,
  author  = {Scalia-Tomba, G},
  title   = {Asymptotic final size distribution of the multitype Reed-Frost process.},
  journal = {J. Appl. Prob.},
  year    = {1986},
  volume  = {23},
  number  = {},
  pages   = {563–584},
  doi     = {10.2307/3213998}
}

@article{Ang,
  author  = {Angeli, L. and Caetano, C.P. and Franco, N. and et al.},
  title   = {Assessing the role of children in the COVID-19 pandemic in Belgium using perturbation analysis.},
  journal = {Nature Communications},
  year    = {2025},
  volume  = {16},
  number  = {},
  pages   = {},
  doi     = {10.1038/s41467-025-57087-z}
}

\section*{Declarations}

\bmhead{Acknowledgements}
T.B.\ is grateful to the Swedish Research Council (grant 2020-0474) for financial support. Most of this work was done while A.B.\ was visiting Stockholm University. A.B.\ is grateful for the hospitality of Stockholm University.

\bmhead{Conflict of interest}
The authors declare no conflict of interest.

\bmhead{Ethics approval}
The paper uses no sensitive data and requires no ethical approval.

\bmhead{Data availability}
The paper uses only publicly available data which is referred to in the article.

\bmhead{Code availability}
Code is available upon request.

\bmhead{Author contribution}
All authors contributed to the study conception and design. The methodological analysis was performed by all authors. The numerical investigations were performed by A.B. and F.B. The first draft of the manuscript was written jointly by all authors and all authors commented on previous versions of the manuscript. All authors read and approved the final manuscript.
\noindent

\begin{appendices}

\section{Proofs of minor results}\label{secA1}
\begin{lemma}\label{lemma:db1}
Let $M\in\R^{k\times k}$ and let $\boldsymbol{\pi}\in(0,1)^k$ with $D_{\boldsymbol{\pi}}:=\diag(\pi_1,\dots,\pi_k)$.
The following are equivalent:
\begin{enumerate}
\item[(i)] $\ \pi_i m_{ij}=\pi_j m_{ji}$ for all $i,j$.
\item[(ii)] $D_{\boldsymbol{\pi}} M$ is symmetric.
\end{enumerate}
Under either condition, the matrix
\[
S\ :=\ D_{\boldsymbol{\pi}}^{1/2}\, M\, D_{\boldsymbol{\pi}}^{-1/2}
\]
is symmetric and similar to $M$, hence $\rho(M)=\rho(S)$.
\end{lemma}
The following proof is based on known results in \cite{Boyd2} where the authors focused on a stochastic matrix instead the next generation matrix.
\begin{proof}
$(i)\Rightarrow(ii)$: $(D_{\boldsymbol{\pi}} M)_{ij}=\pi_i m_{ij}=\pi_j m_{ji}=(D_{\boldsymbol{\pi}} M)_{ji}$.
$(ii)\Rightarrow(i)$ is the same identity read backwards.

If $D_{\boldsymbol{\pi}} M$ is symmetric, then $M^\top D_{\boldsymbol{\pi}}=D_{\boldsymbol{\pi}} M$, i.e.\ $M^\top=D_{\boldsymbol{\pi}} M D_{\boldsymbol{\pi}}^{-1}$.
Thus
\[
S^\top=(D_{\boldsymbol{\pi}}^{1/2} M D_{\boldsymbol{\pi}}^{-1/2})^\top
=D_{\boldsymbol{\pi}}^{-1/2} M^\top D_{\boldsymbol{\pi}}^{1/2}
=D_{\boldsymbol{\pi}}^{-1/2}(D_{\boldsymbol{\pi}} M D_{\boldsymbol{\pi}}^{-1})D_{\boldsymbol{\pi}}^{1/2}
=D_{\boldsymbol{\pi}}^{1/2} M D_{\boldsymbol{\pi}}^{-1/2}=S,
\]
so $S$ is symmetric. Finally, $M=D_{\boldsymbol{\pi}}^{-1/2} S D_{\boldsymbol{\pi}}^{1/2}$ shows that $M$ and $S$ are similar and therefore share the same spectrum, then  $\rho(M)=\rho(S)=\lambda_{max}(S)$, where the last inequality is satisfied because $S$ is symmetric and then for the Spectral Theorem the eigenvalues are real. 
\end{proof}
The symmetric representative $S$ can be interpreted as the \textit{balanced transmission operator}; minimizing $\rho(M)$
becomes minimizing $\rho(S)$ under linear constraints.
\begin{lemma}\label{lemma:bounds}
Define for each $i=1,...,k$
\[
\zeta_i(w)\ :=\ \pi_i\left[1-\exp\left(-\,\displaystyle\frac{w}{\pi_i}\right)\right],\qquad w\ge 0.
\]
For $\boldsymbol{u}\in\mathbb R_{\ge 0}^k$ with $\sum_{i=1}^k u_i=U$, set
\[
F(\boldsymbol{u})\ :=\ \sum_{i=1}^k \zeta_i(u_i).
\]
Then:
\begin{enumerate}
\item[\emph{(i)}] $F(\boldsymbol{u})\ \ge\ \displaystyle\min_{1\le i\le k}\ \zeta_i(U)$.
\item[\emph{(ii)}] Equality in \emph{(i)} holds if and only if $\boldsymbol{y}=U\boldsymbol{e_j}$ for some $j\in\{1,\dots,k\}$.
\end{enumerate}
\end{lemma}
\begin{proof}
For each $i$, $\zeta_i$ is strictly concave on $[0,\infty)$ since
\[
\zeta_i''(w)\ =\ -\frac{1}{\pi_i}\,\rm{e}^{-\,w/\pi_i}\ <\ 0.
\]
Write $z_i:=u_i/\pi_i$ and $f(z):=1-\rm{e}^{-z}$, which is concave on $[0,\infty)$. Then
\[
F(\boldsymbol{u})\ =\ \sum_{i=1}^k \pi_i\, f(z_i).
\]
The feasible set
\[
\mathcal U\ :=\ \Big\{\,\boldsymbol{u}\in\mathbb R_{\ge 0}^k:\ \sum_{i=1}^k u_i=U\,\Big\}
\]
is a simplex (convex and compact). Since $F$ is concave, its minimum over a convex compact set is attained at an extreme point (Bauer maximum principle \cite{Bauer1958}). The extreme points of $\mathcal U$ are precisely $\boldsymbol{u^{(j)}}=U\boldsymbol{e_j}$, $j=1,\dots,k$, and at such a point
\[
F(u^{(j)})\ =\ \zeta_j(U).
\]
        Therefore $F(\boldsymbol{u})\ge \min_j \zeta_j(U)$, with equality if and only if $\boldsymbol{u}=U\,\boldsymbol{e_j}$ for some $j$.
\end{proof}
\begin{lemma}\label{lem:opt-qjj}
Fix $j$ and $r_j>0$. For each $q_{jj}\in[0,1]$, define
\[
f(q_{jj}) \;:=\; (1-q_{jj})\,t_{q_{jj}r_j}.
\]
Then:
\begin{enumerate}
\item [(i)] If $q_{jj}r_j\le 1$ then $f(q_{jj})\equiv 0$.
\item [(ii)] If $q_{jj}r_j>1$, the function $f$ is maximized at
\[
q_{jj}^\star \;=\; \frac{\log(r_j)}{r_j-1}.
\]
\end{enumerate}
\end{lemma}
\begin{proof}
If $q_{jj}r_j\le 1$, the scalar final-size equation (\ref{eq:t_alpha-gamma}) has the unique solution $t_{q_{jj}r_j}=0$, so $f(q_{jj})=0$. 

\medskip
Now assume $q_{jj}r_j>1$, so that $r_j>1$ and $t_{q_{jj}r_j}>0$. We rewrite \eqref{eq:t_alpha-gamma} as
\[
H(q_{jj}):=F(q_{jj},t_{q_{jj}r_j})
\;:=\;
t_{q_{jj}r_j}-1+\exp(-q_{jj}r_j t_{q_{jj}r_j})
\;=\;0.
\]
Differentiating implicitly with respect to $q_{jj}$ gives
\[
0
= \frac{d}{dq_{jj}}H(q_{jj})
= \frac{\partial F}{\partial t_{q_{jj}r_j}}\frac{dt_{q_{jj}r_j}}{dq_{jj}}
+ \frac{\partial F}{\partial q_{jj}},
\]
so that
\[
\frac{dt_{q_{jj}r_j}}{dq_{jj}}
= -\frac{\partial F/\partial q_{jj}}{\partial F/\partial t_{q_{jj}r_j}}.
\]
Since $1-t_{q_{jj}r_j}=\exp(-q_{jj}r_j t_{q_{jj}r_j})$, we obtain
\[
\frac{\partial F}{\partial q_{jj}}
= -t_{q_{jj}r_j}r_j\,\exp(-q_{jj}r_j t_{q_{jj}r_j})
= -t_{q_{jj}r_j}r_j(1-t_{q_{jj}r_j}),
\]
and
\[
\frac{\partial F}{\partial t_{q_{jj}r_j}}
= 1-q_{jj}r_j\,\exp(-q_{jj}r_j t_{q_{jj}r_j})
= 1-q_{jj}r_j(1-t_{q_{jj}r_j}).
\]
Hence
\begin{equation}\label{eq:dt-dq}
\frac{dt_{q_{jj}r_j}}{dq_{jj}}
= \frac{t_{q_{jj}r_j}r_j(1-t_{q_{jj}r_j})}{1-q_{jj}r_j(1-t_{q_{jj}r_j})}.
\end{equation}
Next, computing the derivative of $f$:
\[
\frac{df(q_{jj})}{dq_{jj}}
= -t_{q_{jj}r_j}+(1-q_{jj})\frac{dt_{q_{jj}r_j}}{dq_{jj}}.
\]
Setting $\displaystyle\frac{df}{dq_{jj}}=0$ and using \eqref{eq:dt-dq} yields
\[
-t_{q_{jj}r_j}
+ (1-q_{jj})\frac{t_{q_{jj}r_j}r_j(1-t_{q_{jj}r_j})}{1-q_{jj}r_j(1-t_{q_{jj}r_j})}
= 0
\quad\text{if and only if}\quad
\frac{-1+r_j-t_{q_{jj}r_j}r_j}{1-q_{jj}r_j(1-t_{q_{jj}r_j})} = 0,
\]
which implies
\[
t_{q_{jj}r_j} = 1-\frac{1}{r_j}.
\]
Substituting this into \eqref{eq:t_alpha-gamma} gives
\[
1-\frac{1}{r_j}
= 1-\exp\left(-q_{jj}r_j\left(1-\frac{1}{r_j}\right)\right)
\,\;\text{if and only if}\;\,
\exp\left(-q_{jj}r_j\left(1-\frac{1}{r_j}\right)\right) = \frac{1}{r_j},
\]
and therefore
\[
q_{jj}^\star
= \frac{\log(r_j)}{r_j-1}>0.
\]
\end{proof}
\begin{lemma}\label{lem:kkt}
Let $\boldsymbol{r}=(r_1,...,r_k)^\top$, $r_i>0$ for all $i$ and 
\[
\psi_i(u_i)=\pi_i\!\left[1-\exp\left(-\displaystyle\frac{u_i}{\pi_i}\right)\right],\qquad 
\Psi(\boldsymbol{u})=\sum_{i=1}^{k}\psi_i(u_i).
\]
Define the affine constraint
\[
h(\boldsymbol{u})=\sum_{i=1}^{k}u_i-\sum_{i=1}^{k}r_i\,\psi_i(u_i)=0,
\]
with \(u_i\ge 0\).
Consider the Lagrangian
\[
\mathcal{L}(\boldsymbol{u},\lambda,\nu)=-\Psi(\boldsymbol{u})+\lambda\,h(\boldsymbol{u})-\sum_{i=1}^k \nu_i u_i.
\]
Then the problem
\[
\max_{\boldsymbol{u}\ge 0}\ \Psi(\boldsymbol{u})\quad\text{s.t.}\quad h(\boldsymbol{u})=0
\]
admits the following solution.

\smallskip
\emph{(a) If \(r_i\le 1\) for all \(i\)}, the unique feasible point is \(u^\star=0\) and \(\max\Psi=0\).

\smallskip
\emph{(b) If there exists \(i\) with \(r_i>1\)}, there exists a unique \(\lambda^\star>0\) such that the optimizer \(u^\star\) is
\[
u_i^\star \;=\; \pi_i\left[\log\!\left(r_i+\displaystyle\frac{1}{\lambda}\right)\right]_+
\quad\quad where\quad[x]_+=\max\{x,0\}.
\]
The scalar \(\lambda^\star\) is the unique solution \(\lambda>0\) of
\[
F(\lambda)=\sum_{i\in\mathcal{A}_\lambda}\!\left[\ \pi_i\,\log\!\left(r_i+\displaystyle\frac{1}{\lambda}\right)\;-\;r_i\,\pi_i\Bigl(1-\frac{\lambda}{1+\lambda r_i}\Bigr)\right]=0.
\]
where $\mathcal{A}_\lambda=\left\{i\in{1,...,k}\quad s.t.\quad r_i+\displaystyle\frac{1}{\lambda}>1 \right\}$, and the optimal value is
\[
\Psi(\boldsymbol{u}^\star)=\sum_{i\in\mathcal{A}_{\lambda^*}} \pi_i\!\left(1-\frac{\lambda^\star}{1+\lambda^\star r_i}\right).
\]
\end{lemma}
\begin{proof}
Each \(\psi_i\) is increasing and strictly concave on \([0,\infty)\) since
\[
\psi_i'(u_i)=\exp\left(-\displaystyle\frac{u_i}{\pi_i}\right)\in(0,1],\qquad
\psi_i''(u_i)=-\frac{1}{\pi_i}\exp\left(-\displaystyle\frac{u_i}{\pi_i}\right)<0.
\]
Hence the linear combination \(\Psi=\displaystyle\sum_{i=1}^k \psi_i\) is strictly concave. Therefore, maximizing the concave objective \(\Psi\) over the convex feasible set defined by the affine equality \(h(\boldsymbol{u})=0\) and the inequalities \(u_i\ge 0\) is a convex optimization problem and to solve it we use the Karush-Kuhn-Tucker conditions (\cite{KKT}, Ch 5.5.3). Under these constraints, the KKT conditions are necessary, however by strict concavity of \(\Psi\) they are also sufficient; consequently, any feasible KKT point is the unique global maximizer (\cite{KKT}, Ch 4.2.1).

\smallskip
The KKT conditions
\begin{subequations}\label{eq:KKT}
\begin{align}
u_i &\ge 0, && i=1,\dots,k, \tag{KKT1}\label{KKT1}\\
h(\boldsymbol{u}) &= 0, \tag{KKT2}\label{KKT2}\\
\nu_i &\ge 0, && i=1,\dots,k, \tag{KKT3}\label{KKT3}\\
u_i\,\nu_i &= 0, && i=1,\dots,k, \tag{KKT4}\label{KKT4}\\
\frac{\partial \mathcal{L}}{\partial u_i}(\boldsymbol{u},\lambda,\nu) &= 0, && i=1,\dots,k. \tag{KKT5}\label{KKT5}
\end{align}
\end{subequations}
By the stationarity condition \eqref{KKT5}, for each \(i\),
\begin{align*}
\frac{\partial\mathcal{L}}{\partial u_i}
&=-\,\psi_i'(u_i)+\lambda\!\left(1-r_i\,\psi_i'(u_i)\right)-\nu_i\\
&=-\,\exp\left(-\displaystyle\frac{u_i}{\pi_i}\right)+\lambda\!\left(1-r_i \exp\left(-\displaystyle\frac{u_i}{\pi_i}\right)\right)-\nu_i=0.
\end{align*}
\noindent\emph{Active coordinates (\(u_i>0\)).}\\
By complementary slackness \eqref{KKT4} we have \(\nu_i=0\), hence
\[
\exp\left(-\displaystyle\frac{u_i}{\pi_i}\right)=\frac{\lambda}{1+\lambda r_i}.
\]
\noindent\emph{Inactive coordinates (\(u_i=0\)).}\\
Stationarity reduces to
\[
-1+\lambda(1-r_i)-\nu_i=0,
\]
with \(\nu_i\ge 0\) by \eqref{KKT3}.

\smallskip
\textbf{Case (a): \(r_i\le 1\)} for all \(i\).\\
For every \(x\ge 0\) we have \(1-\exp\left(-x\right)\le x\) with equality iff \(x=0\).
Hence \(\psi_i(u_i)\le u_i\) for all \(i\), and thus
\[
\sum_{i=1}^k r_i\,\psi_i(u_i)\ \le\ \sum_{i=1}^k \psi_i(u_i)\ \le\ \sum_{i=1}^k u_i.
\]
Imposing the equality constraint \eqref{KKT2} (i.e., \(h(\boldsymbol{u})=0\)) gives
\[
\sum_{i=1}^k u_i \ =\ \sum_{i=1}^k r_i\,\psi_i(u_i)\ \le\ \sum_{i=1}^k \psi_i(u_i)\ \le\ \sum_{i=1}^k u_i,
\]
so all inequalities are equalities. Therefore \(\psi_i(u_i)=u_i\) for every \(i\),
which by $1-\exp\left(-x\right)=x$ if and only if $x=0$ forces \(u_i=0\) for all \(i\).
In particular, the only feasible point satisfying \eqref{KKT1}-\eqref{KKT2} is
\(\boldsymbol{u}^\star=0\), whence \(\Psi(\boldsymbol{u}^\star)=0\). By strict concavity of \(\Psi\),
\(\boldsymbol{u}^\star\) is the unique global maximizer and it trivially satisfies
\eqref{KKT1}-\eqref{KKT5}.

\smallskip
\textbf{Case (b):} there exists \(i\) with \(r_i>1\).\\
Assume \(\lambda>0\) and let \(\mathcal{A}=\{i:\,u_i>0\}\) denote the active set.

\medskip
\emph{Inactive coordinates} (\(i\notin\mathcal{A}\), i.e., \(u_i=0\) by \eqref{KKT1}).
From stationarity \eqref{KKT5} we get
\[
\nu_i \;=\; \lambda(1-r_i)-1.
\]
Feasibility of the multiplier \(\nu_i\ge 0\) by \eqref{KKT3} requires
\[
\lambda(1-r_i)\ \ge\ 1\quad\Longleftrightarrow\quad
\lambda \ \ge\ \frac{1}{\,1-r_i\,}.
\]
This inequality can hold only if \(r_i<1\). If \(r_i\ge 1\), the left-hand side is non positive 
\((\le 0)\) for \(\lambda>0\), hence the constraint \(\nu_i\ge 0\) cannot be met.
Therefore, any index with \(r_i\ge 1\) cannot be inactive when \(\lambda>0\).

\medskip
\emph{Active coordinates} (\(i\in\mathcal{A}\), i.e., \(u_i>0\)).
By complementary slackness \eqref{KKT4} we have \(\nu_i=0\), and \eqref{KKT5} yields
\[
\exp\left(-\displaystyle\frac{u_i}{\pi_i}\right)=\frac{\lambda}{1+\lambda r_i}\in(0,1)
\qquad\text{if and only if}\qquad
u_i=\pi_i\log\!\left(\frac{1+\lambda r_i}{\lambda}\right)>0.
\]
For \(r_i\ge 1\) this positivity holds automatically for any \(\lambda>0\).
For \(r_i<1\), the condition \(u_i>0\) is equivalent to
\(\lambda<\displaystyle\frac{1}{1-r_i}\).

\medskip
Summarizing, for each \(i\):
\[
\begin{cases}
r_i\ge 1 \quad\text{then}\quad\quad u_i>0,\\[3pt]
r_i<1 \quad\text{then}\quad 
\begin{cases}
u_i>0\ \text{ if } \lambda<\dfrac{1}{1-r_i},\\[6pt]
u_i=0\ \text{ if } \lambda\ge\dfrac{1}{1-r_i}.
\end{cases}
\end{cases}
\]
The multiplier \(\lambda>0\) is then determined uniquely by the equality constraint
\eqref{KKT2}, i.e., by \(h(\boldsymbol{u})=0\). We define thus
\[
\mathcal{A}=\mathcal{A}_\lambda:=\left\{\,i:\ r_i+\displaystyle\frac{1}{\lambda}>1\,\right\}.
\]

\smallskip
For \(i\in\mathcal{A}_\lambda\) we have
\[
u_i=\pi_i\log\!\left(\frac{1+\lambda r_i}{\lambda}\right)
\quad\text{and}\quad
\psi_i(u_i)=\pi_i\!\left(1-\frac{\lambda}{1+\lambda r_i}\right).
\]
Plugging the explicit \(u_i\) and \(\psi_i(u_i)\) into the constraint \(h(\boldsymbol{u})=0\) gives
\[
F(\lambda):=\sum_{i\in\mathcal{A}_\lambda}
\left[\ \pi_i\log\!\left(r_i+\frac{1}{\lambda}\right)
\;-\;r_i\,\pi_i\left(1-\frac{\lambda}{1+\lambda r_i}\right)\right]=0.
\]
We next show that \(F\) is strictly decreasing on \((0,\infty)\) and crosses zero exactly once.

\smallskip
For fixed \(i\) with \(r_i+1/\lambda>1\),
\[
\displaystyle\frac{d}{d\lambda}\,\log\!\left(r_i+1/\lambda\right)
=-\frac{1}{\lambda(1+\lambda r_i)},
\qquad
\displaystyle\frac{d}{d\lambda}\left[-r_i\left(1-\displaystyle\frac{\lambda}{1+\lambda r_i}\right)\right]
=\displaystyle\frac{r_i}{(1+\lambda r_i)^2}.
\]
Hence
\begin{align*}
    \frac{d}{d\lambda}\left[\pi_i\log\!\left(r_i+1/\lambda\right)
-r_i\pi_i\left(1-\frac{\lambda}{1+\lambda r_i}\right)\right]&=\pi_i\!\left[-\frac{1}{\lambda(1+\lambda r_i)}+\frac{r_i}{(1+\lambda r_i)^2}\right]\\
&=-\frac{\pi_i}{\lambda(1+\lambda r_i)^2}.
\end{align*}
Therefore,
\[
F'(\lambda)=-\sum_{i\in\mathcal{A}_\lambda}\frac{\pi_i}{\lambda(1+\lambda r_i)^2}<0,
\]
so \(F\) is strictly decreasing on \((0,\infty)\).

\smallskip
As $\lambda$ tends to $ 0^+$,
\[
\log\!\left(r_i+\displaystyle\frac{1}{\lambda}\right)
=\log\!\left(\displaystyle\frac{1+\lambda r_i}{\lambda}\right)
=-\log \lambda+\log(1+\lambda r_i)
=-\log \lambda+O(1)
\sim \log\!\left(\tfrac{1}{\lambda}\right)\to +\infty,
\]
while the second term remains \(O(1)\); hence \(F(\lambda)\to +\infty\).
As $\lambda$ tends to $ +\infty$,
\[
\log\!\left(r_i+\tfrac{1}{\lambda}\right)\ \longrightarrow\ \log r_i,
\]
and
\[
1-\frac{\lambda}{1+\lambda r_i}\longrightarrow 1-\frac{1}{r_i}=\frac{r_i-1}{r_i}.
\]
Hence
\[
\lim_{\lambda\to +\infty}F(\lambda)=\sum_{i:\,r_i>1}\pi_i\bigl(\log r_i-(r_i-1)\bigr)<0,
\]
since \(\log x\le x-1\) with strict inequality for \(x\ne1\).
By the intermediate value theorem and strict monotonicity, there exists a unique \(\lambda^\star>0\) with \(F(\lambda^\star)=0\).

\smallskip
Set
\[
u_i^\star=\pi_i\,\left[\log\!\left(r_i+\displaystyle\frac{1}{\lambda}\right)\right]_+,
\]
so that
\[
u_i^\star>0\quad\text{if and only if}\quad r_i+\displaystyle\frac{1}{\lambda^\star}>1.
\]
Define
\[
\nu_i^\star=
\begin{cases}
0, &\text{if } u_i^\star>0,\\[2pt]
\lambda^\star(1-r_i)-1, &\text{if } u_i^\star=0,
\end{cases}
\]
which satisfies \(\nu_i^\star\ge0\) on the inactive indices as argued above. Then all KKT conditions hold and the optimal value is
\[
\Psi(\boldsymbol{u}^\star)=\sum_{i\in\mathcal{A}_{\lambda^\star}} \pi_i\!\left[1-\frac{\lambda^\star}{1+\lambda^\star r_i}\right]=\sum_{i=1}^k \pi_i\!\left[1-\frac{\lambda^\star}{1+\lambda^\star r_i}\right]_+.
\]
\end{proof}
We note that no KKT solution can have \(\lambda\le 0\).

\medskip
The case \(\lambda=0\) is impossible, in fact for any inactive coordinate (i.e. \(u_i=0\)) by \eqref{KKT1}, stationarity \eqref{KKT5} gives
\[
\nu_i=\lambda(1-r_i)-1=-1,
\]
which violates dual feasibility \eqref{KKT3}. For any active coordinate
(i.e. \(u_i>0\)), complementary slackness \eqref{KKT4} yields \(\nu_i=0\) and
\eqref{KKT5} becomes \(\rm{e}^{-u_i/\pi_i}=\lambda/(1+\lambda r_i)=0\), which is
impossible since $\rm{e}^{-u_i/\pi_i}\in(0,1]$. Hence \(\lambda=0\) cannot occur.

\medskip
If there exists at least one \(r_i\le 1\), for an active index (i.e. \(u_i>0\)), \eqref{KKT4} implies \(\nu_i=0\), and
\eqref{KKT5} gives
\[
\exp\left(-\displaystyle\frac{u_i}{\pi_i}\right)=\displaystyle\frac{\lambda}{1+\lambda r_i}\in(0,1].
\]
If \(\lambda<0\), the right-hand side is \(>0\) only when \(1+\lambda r_i<0\),
i.e.\ \(\lambda<-1/r_i\). In particular, if \(r_i\le 1\), then \(1+\lambda r_i>0\) for every \(\lambda<0\), hence
\(\lambda/(1+\lambda r_i)\le 0\), a contradiction. Thus indices with \(r_i\le 1\)
cannot be active when \(\lambda<0\). If \(u_i=0\) (inactive), \eqref{KKT5} gives \(\nu_i=\lambda(1-r_i)-1\).
If \(r_i<1\) and \(\lambda<0\), then \(\lambda(1-r_i)<0\), so \(\nu_i<-1<0\),
violating \eqref{KKT3}.
Therefore, as soon as there exists at least one \(r_i\le 1\), no KKT point solution can
have \(\lambda<0\).

\medskip
If \(r_i>1\ \forall i\) also excludes \(\lambda<0\). In fact, set \(z_i:=\exp(-u_i/\pi_i)\). For active indices,
\eqref{KKT4}-\eqref{KKT5} give
\[
z_i=\frac{\lambda}{1+\lambda r_i},\qquad
\lambda=\frac{z_i}{1-r_i z_i}.
\]
If \(\lambda<0\) and \(r_i>1\), necessarily \(1-r_i z_i<0\), hence
\(z_i>1/r_i\), i.e.\ \(z_i\in(1/r_i,1)\).
Rewriting the equality constraint \eqref{KKT2} in terms of \(z_i\):
\[
0=\sum_{i=1}^k u_i-\sum_{i=1}^k r_i\psi_i(u_i)
=\sum_{i=1}^k \pi_i\bigl[-\log z_i - r_i(1-z_i)\bigr]
=\sum_{i=1}^k \pi_i\,g_i(z_i),
\]
where \(g_i(z):=-\log z - r_i(1-z)\) for \(z\in(0,1]\).
We have \(g_i'(z)=-1/z+r_i\) and \(g_i''(z)=1/z^2>0\), so each \(g_i\) is convex
with a unique minimum at \(z=1/r_i\), where
\[
g_i(1/r_i)=\log r_i-(r_i-1)<0 \qquad (r_i>1).
\]
Moreover, on \((1/r_i,1)\) one has \(g_i'(z)>0\), hence
\(g_i(z)\in\bigl(\log r_i-r_i+1,\,0\bigr)\subset(-\infty,0)\).
Thus, if \(\lambda<0\) forcing every active \(z_i\in(1/r_i,1)\),
\[
\sum_{i\in \mathcal{A}}\pi_i\,g_i(z_i)\;<\;0.
\]
Inactive indices have \(u_i=0\) and contribute \(r_i\psi_i(0)=0\), so they cannot
offset the negativity. Therefore \(\sum_i \pi_i\,\phi_i(z_i)<0\), contradicting
\eqref{KKT2}. Hence no solution with \(\lambda<0\) exists when \(r_i>1\) for all \(i\).

\medskip
In conclusion:
\(\lambda=0\) violates \eqref{KKT3}-\eqref{KKT5}, and \(\lambda<0\) contradicts
\eqref{KKT1}-\eqref{KKT5} either if there exist some \(r_i\le 1\) or if \(r_i>1\ \forall i\). Consequently, the
closing equation \(F(\lambda)=0\) admits a unique solution with \(\lambda^\star>0\).
\begin{lemma}\label{lemma:lmi}
Let $S=\{s_{i,j}\}\in\mathbb{R}^{k\times k}$ a symmetric matrix ($S=S^\top$) and $\lambda_{max}(S)$ its maximum eigenvalue, then:
\begin{enumerate}
\item[(i)] Its epigraph has the linear matrix inequality representation
\[
\bigl\{ (S,t): \lambda_{max}(S)\le t \bigr\}
\ =\
\bigl\{ (S,t):\ tI - S \succeq 0 \bigr\}.
\]
\item[(ii)] $S\mapsto \lambda_{max}(S)$ is convex.
\end{enumerate}
In particular, if additionally $S\geq 0$ entrywise, then $\rho(S)=\lambda_{max}(S)$ and therefore:
\[
\bigl\{ (S,t): \rho(S)\le t \bigr\}
\ =\
\bigl\{ (S,t):\ tI - S \succeq 0 \bigr\}.
\]
is a convex space.
\end{lemma}
\begin{proof}
(i) Suppose $tI - S \succeq 0$. Then for every $\boldsymbol{v}\in\R^k$,
\[
\boldsymbol{v}^\top (tI - S)\boldsymbol{v} \ge 0 \quad\text{if and only if}\quad \boldsymbol{v}^\top S \boldsymbol{v} \le t\|\boldsymbol{v}\|^2.
\]
In particular, for $\|\boldsymbol{v}\|=1$ we obtain $\boldsymbol{v}^\top S \boldsymbol{v} \le t$, hence $\lambda_{max}(S)\le t$.

\medskip
(ii) Since $S$ is symmetric, the spectral radius can be expressed using the Rayleigh quotient (see \cite{RR}),
\[
\lambda_{max}(S) \;=\; \max_{\|\boldsymbol{v}\|=1} \boldsymbol{v}^\top S \boldsymbol{v}.
\]
Moreover, for each vector $\boldsymbol{v}\in\mathbb{R}^k$,
\begin{equation*}
tr(S\boldsymbol{v}\boldsymbol{v}^\top)=\displaystyle\sum_{i=1}^{k}\left(S\boldsymbol{v}\boldsymbol{v}^\top \right)_{ii}=\displaystyle\sum_{i=1}^{k}\displaystyle\sum_{j=1}^{k}s_{ij}(\boldsymbol{v}\boldsymbol{v}^\top)_{ji}=\displaystyle\sum_{i=1}^{k}\displaystyle\sum_{j=1}^{k}s_{ij}v_iv_j=\displaystyle\sum_{i=1}^{k}\displaystyle\sum_{j=1}^{k}v_is_{ij}v_j=\boldsymbol{v}^\top S \boldsymbol{v}.
\end{equation*}
Since the trace is linear, the map $S\mapsto \boldsymbol{v}^\top S \boldsymbol{v} = tr(S \boldsymbol{v}\boldsymbol{v}^\top)$ is linear in $S$ and since the maximum of linear (hence convex) functions is convex, it follows that $\lambda_{max}(S)$ is convex.

\medskip
Conversely, if $\lambda_{max}(S)\le t$, then all eigenvalues of $S$ are at most $t$.  
Therefore all eigenvalues of $tI-S$ are nonnegative, i.e.\ $tI - S\succeq 0$.

\medskip
Furthermore, if $S\geq 0$ entrywise, from Perron-Frobenius's Theorem, the maximum eigenvalue is real, non-negative and it corresponds to the spectral radius of $S$.
This proves the equivalence
\[
\rho(S)\le t \quad\text{if and only if}\quad tI - S\succeq 0,
\]
and from (i) the spectral radius $\rho(S)$ is convex hence also its epigraph. This claimed linear matrix inequality representation of the epigraph.
\end{proof}
\begin{lemma}\label{lem:cont-eps}
Let $\boldsymbol{\tau}(M)\in[0,1]^k$ denote the maximal fixed point of
\begin{equation*}
\boldsymbol{\tau}
=T(M,\boldsymbol{\tau})
=\boldsymbol{1}-\exp\!\Bigl(-D_{\boldsymbol{\pi}}^{-1}M^\top D_{\boldsymbol{\pi}}\,\boldsymbol{\tau}\Bigr),
\end{equation*}
where the exponential is taken componentwise, and define the total final size
$\bar\tau(M):=\boldsymbol{\pi}^\top\boldsymbol{\tau}(M)$.
Let $(M_\varepsilon)_{\varepsilon>0}\subset\mathbb{R}^{k\times k}_{\ge 0}$ be such that
$M_\varepsilon\to M$ entrywise as $\varepsilon\to 0$. Then:
\begin{enumerate}
\item[i)] $\rho(M_\varepsilon)\to \rho(M)$.
\item[ii)] $\boldsymbol{\tau}(M_\varepsilon)\to \boldsymbol{\tau}(M)$ componentwise.
\item[iii)] $\bar\tau(M_\varepsilon)\to \bar\tau(M)$.
\end{enumerate}
\end{lemma}
\begin{proof}
\noindent
i) The continuity of $\rho(\cdot)$ follows from the continuity of eigenvalues with respect to the matrix entries,
see Theorem 2.4.9.2 in \cite{RR}.

\medskip
\noindent
ii) Fix $M$ and define the iterations
\[
\boldsymbol{x}_M^{(0)}:=\boldsymbol{1},
\qquad
\boldsymbol{x}_M^{(n+1)}:=T\!\bigl(M,\boldsymbol{x}_M^{(n)}\bigr),
\qquad n\ge 0.
\]
The map $F(M,\cdot)$ is order-preserving on $[0,1]^k$: if $\boldsymbol{x}\le \boldsymbol{y}$ componentwise, then
$F(M,\boldsymbol{x})\le T(M,\boldsymbol{y})$ componentwise, because
$D_{\boldsymbol{\pi}}^{-1}M^\top D_{\boldsymbol{\pi}}\ge 0$ and the componentwise function $z\mapsto 1-\rm{e}^{-z}$ is increasing.
Since $T(M,\boldsymbol{1})\le \boldsymbol{1}$, the sequence $(\boldsymbol{x}_M^{(n)})_{n\ge 0}$ is componentwise non-increasing,
hence it converges componentwise to some $\boldsymbol{x}_M^\infty\in[0,1]^k$.
By continuity of $T(M,\cdot)$, $\boldsymbol{x}_M^\infty$ is a fixed point. Moreover, if $\boldsymbol{y}\in[0,1]^k$ is any fixed point,
then $\boldsymbol{y}\le \boldsymbol{x}_M^{(0)}$ and monotonicity implies $\boldsymbol{y}\le \boldsymbol{x}_M^{(n)}$ for all $n$,
hence $\boldsymbol{y}\le \boldsymbol{x}_M^\infty$. Therefore $\boldsymbol{x}_M^\infty$ is the maximal fixed point, i.e.\
$\boldsymbol{x}_M^\infty=\boldsymbol{\tau}(M)$.
Assume now that $M_\varepsilon\ge M$ entrywise and $M_\varepsilon\to M$ entrywise as $\varepsilon\to 0$. Then for every $\boldsymbol{x}\in[0,1]^k$ we have
$T(M_\varepsilon,\boldsymbol{x})\ge T(M,\boldsymbol{x})$ componentwise, hence
\[
\boldsymbol{x}_{M_\varepsilon}^{(n)} \ge \boldsymbol{x}_M^{(n)} \qquad\text{for all }n,
\]
and in particular $\boldsymbol{\tau}(M_\varepsilon)\ge \boldsymbol{\tau}(M)$.
Fix $\delta>0$. Since $\boldsymbol{x}_M^{(n)}\to \boldsymbol{\tau}(M)$ componentwise, choose $n_\delta$ such that
\[
\boldsymbol{0}\le \boldsymbol{x}_M^{(n_\delta)}-\boldsymbol{\tau}(M)\le \delta\,\boldsymbol{1}
\qquad\text{componentwise.}
\]
For any fixed $\boldsymbol{x}\in[0,1]^k$, the map $M\mapsto T(M,\boldsymbol{x})$ is continuous (entrywise),
since $M\mapsto D_{\boldsymbol{\pi}}^{-1}M^\top D_{\boldsymbol{\pi}}\,\boldsymbol{x}$ is linear in $M$ and the componentwise
exponential is continuous. Therefore, by induction on $n$, for each fixed $n\ge 0$ we have
\[
\boldsymbol{x}_{M_\varepsilon}^{(n)} \longrightarrow \boldsymbol{x}_M^{(n)}
\qquad\text{componentwise as }\varepsilon\to 0.
\]
In particular, this holds for $n=n_\delta$, hence there exists $\varepsilon_\delta>0$ such that for all $\varepsilon\in(0,\varepsilon_\delta)$,
\[
\bigl|\boldsymbol{x}_{M_\varepsilon}^{(n_\delta)}-\boldsymbol{x}_M^{(n_\delta)}\bigr|\le \delta\,\boldsymbol{1}
\qquad\text{componentwise.}
\]
Since $\boldsymbol{\tau}(M_\varepsilon)\le \boldsymbol{x}_{M_\varepsilon}^{(n_\delta)}$ componentwise, we obtain
\[
\boldsymbol{0}\le \boldsymbol{\tau}(M_\varepsilon)-\boldsymbol{\tau}(M)
\le \boldsymbol{x}_{M_\varepsilon}^{(n_\delta)}-\boldsymbol{\tau}(M)
= \bigl(\boldsymbol{x}_{M_\varepsilon}^{(n_\delta)}-\boldsymbol{x}_M^{(n_\delta)}\bigr)
   +\bigl(\boldsymbol{x}_M^{(n_\delta)}-\boldsymbol{\tau}(M)\bigr),
\]
and therefore
\[
\bigl|\boldsymbol{\tau}(M_\varepsilon)-\boldsymbol{\tau}(M)\bigr|
\le 2\delta\,\boldsymbol{1}
\qquad\text{componentwise,}
\]
for all $\varepsilon$ small enough. \\This proves $\boldsymbol{\tau}(M_\varepsilon)\to \boldsymbol{\tau}(M)$ componentwise.

\medskip
\noindent
iii) This follows immediately from (ii) since $\bar\tau(M)=\boldsymbol{\pi}^\top\boldsymbol{\tau}(M)$ is a linear functional
with fixed coefficients $\boldsymbol{\pi}$.
\end{proof}

\begin{lemma}
\label{lem:familyofQ}
Let $\boldsymbol{s}\in\mathbb{R}_+^k$ and $\boldsymbol{u}\in\mathbb{R}_+^k$ satisfy
\[
\sum_{i=1}^k s_i=\sum_{i=1}^k u_i,
\qquad s_k>0,
\]
and consider the constraints
\begin{equation}\label{eq:constraints_Q_lem}
Q\mathbf 1=\mathbf 1,
\qquad
Q^\top \boldsymbol{s}=\boldsymbol{u}.
\end{equation}
Then every $Q\in\mathbb{R}^{k\times k}$ satisfying \eqref{eq:constraints_Q_lem} is uniquely determined by its
upper-left $(k-1)\times(k-1)$ block $\{q_{ij}\}_{i,j\le k-1}$.
Conversely, for any matrix $E\in\mathbb{R}^{(k-1)\times(k-1)}$ satisfying
\begin{align}
&e_{ij}\ge 0, && i,j\le k-1, \label{eq:cond_lem}\\
&\sum_{j=1}^{k-1}e_{ij}\le 1, && i\le k-1, \nonumber\\
&\sum_{i=1}^{k-1}s_ie_{ij}\le u_j, && j\le k-1, \nonumber\\
&\sum_{i=1}^{k-1}u_i -\sum_{i=1}^{k-1}\sum_{j=1}^{k-1}s_i e_{ij} \le s_k .\nonumber
\end{align} 
define $Q(E)$ by
\begin{align}
q_{ij}&=e_{ij}, && i,j\le k-1, \nonumber\\
q_{ik}&=1-\sum_{j=1}^{k-1}e_{ij}, && i\le k-1, \label{eq:QE_def_compact}\\
q_{kj}&=\frac{u_j-\sum_{i=1}^{k-1}s_ie_{ij}}{s_k}, && j\le k-1, \nonumber\\
q_{kk}&=1-\sum_{j=1}^{k-1}q_{kj}. \nonumber
\end{align}
Then $Q(E)$ always satisfies \eqref{eq:constraints_Q_lem}, and every solution of \eqref{eq:constraints_Q_lem}
is of the form $Q(E)$.
\end{lemma}
\begin{proof}
Fix any $E\in\mathbb{R}^{(k-1)\times(k-1)}$ and define $Q(E)$ by \eqref{eq:QE_def_compact}.
By construction it is clear that each row of $Q(E)$ sums to one, hence $Q(E)\mathbf 1=\mathbf 1$.
Moreover, for every $j\le k-1$,
\[
(Q(E)^\top \boldsymbol{s})_j
=\sum_{i=1}^{k-1}s_i q_{ij}+s_k q_{kj}
=\sum_{i=1}^{k-1}s_ie_{ij}+s_k\frac{u_j-\sum_{i=1}^{k-1}s_ie_{ij}}{s_k}
=u_j.
\]
Since the sum of all components $\mathbf 1^\top(Q(E)^\top\boldsymbol{s})=(Q(E)\mathbf 1)^\top\boldsymbol{s}=\mathbf 1^\top\boldsymbol{s}
=\mathbf 1^\top\boldsymbol{u}$, the $k$-th component also matches, hence $Q(E)^\top\boldsymbol{s}=\boldsymbol{u}$ and
\eqref{eq:constraints_Q_lem} holds.

Conversely, let $Q$ satisfy \eqref{eq:constraints_Q_lem} and set $E=(q_{ij})_{i,j\le k-1}$.
Then $q_{ik}=1-\sum_{j=1}^{k-1}e_{ij}$ follows from $Q\mathbf 1=\mathbf 1$ for $i\le k-1$.
For $j\le k-1$, the constraint $(Q^\top\boldsymbol{s})_j=u_j$ yields
\[
q_{kj}=\frac{u_j-\sum_{i=1}^{k-1}s_ie_{ij}}{s_k},
\]
and $q_{kk}=1-\sum_{j=1}^{k-1}q_{kj}$ follows from the $k$-th row-sum constraint.
Therefore $Q=Q(E)$, proving both existence and uniqueness of the parametrization.

Finally, $Q(E)\ge 0$ is equivalent to the nonnegativity of its entries. Using \eqref{eq:QE_def_compact},
this is exactly the system \eqref{eq:cond_lem}.
\end{proof}
Note that in the proof of Theorem \ref{thm:tot-fs-row} has been chosen $Q^\star=Q(E)$ where $e_{ij}=q^\star_j=u_j/\sum_{i=1}^ku_i$.

\section{Computation of detailed balance bounds}
\subsection{Recipe to find sharp lower bound for $R_0$ under Detailed Balance}\label{recipe}
Fix $\boldsymbol{\pi}$ and either the column sums $\boldsymbol{c}$ or the row sums $\boldsymbol{r}$. In principle, to bound the reproduction number $R_0$ one must explore the entire set of nonnegative matrices \(M\) satisfying detailed balance \(D_{\boldsymbol{\pi}} M=(D_{\boldsymbol{\pi}} M)^\top\) and compute their spectral radius \(\rho(M)\). A naive numerical strategy discretizes each free entry on an \(m\)-point grid, producing \(m^{\,k(k-1)}\) candidates already prohibitive for \(k>2\) even with \(m=10\). We instead propose a constructive recipe and prove that this task is exactly equivalent to a convex semidefinite program, for which efficient algorithms and software exist, for more details see e.g. \cite{SDP}.

\medskip\noindent
\textbf{Step 1.}  
Define
\[
S \;:=\; D_{\boldsymbol{\pi}}^{1/2} M D_{\boldsymbol{\pi}}^{-1/2}.
\]
From Lemma \ref{lemma:db1}, $S$ is similar to $M$ hence
\[
\rho(M)=\rho(S).
\]
\textbf{Step 2.}  
For each column $j=1,...,k$,
\begin{equation}
c_j = \sum_{i=1}^k m_{ij}
= \sum_{i=1}^k \left( s_{ij}\,\frac{\sqrt{\pi_j}}{\sqrt{\pi_i}} \right).
\label{eq_c_constrain}
\end{equation}
Define 
\[
w_i:=1/\sqrt{\pi_i}\quad \text{and}\quad \tilde c_j:=c_j/\sqrt{\pi_j}.
\]
Then, recalling that S is symmetric for Lemma \ref{lemma:db1}, the above condition in \eqref{eq_c_constrain} is equivalent to
\[
(S\boldsymbol{w})_j = \tilde c_j, \qquad \text{for all $j=1,..,k$}.
\]
Hence, the column constraints in the vector form become
\begin{equation}
\label{aff.constrain}
S\boldsymbol{w} = \tilde{\boldsymbol{c}}.    
\end{equation}
Nonnegativity of $M$ translates to entrywise nonnegativity of $S$:
\[
m_{ij} = s_{ij}\frac{\sqrt{\pi_j}}{\sqrt{\pi_i}} \ge 0
\quad\Longleftrightarrow\quad s_{ij}\ge 0.
\]
\medskip\noindent
\textbf{Step 3.}  
The feasible set in $S$-elements is
\[
\mathcal{S}_c \;=\; \{\,S\in\R^{k\times k} : S=S^\top,\ S\ge 0,\ S\boldsymbol{w}=\tilde{\boldsymbol{c}} \,\}.
\]
On this set we want to minimize $\rho(S)$.
For Lemma \ref{lemma:lmi} the minimization problem is equivalent to the convex semidefinite program below.
\begin{equation}\label{eq:lsp1}
\begin{aligned}
\min_{S,t}\quad & t \\
\text{s.t.}\quad & S\ge 0, \\
& S=S^\top, \\
& S \boldsymbol{w} = \tilde{\boldsymbol{c}}, \\
& t I - S \geq 0,
\end{aligned}
\end{equation}
Note that all the constraints represent convex spaces and therefore the intersection is a convex space.

\medskip\noindent
\textbf{Step 4.}  
The feasible set $\mathcal{S}_c$ is closed.  
It is also bounded: indeed, from $S\boldsymbol{w}=\tilde{\boldsymbol{c}}$ with $w_i>0$ we have
\[
\sum_{i=1}^k w_i s_{ij} = \tilde c_j \quad\text{implies}\quad
0 \le s_{ij} \le \frac{\tilde c_j}{w_i}.
\]
Thus each entry of $S$ is bounded.  
Hence $\mathcal{S}_c$ is compact.  
Since $\rho(S)$ is continuous, the minimum over $\mathcal{S}_c$ exists and it is attained.  
This proves existence of an optimizer $S^\star$ and hence $M^\star=D_{\boldsymbol{\pi}}^{-1/2}S^\star D_{\boldsymbol{\pi}}^{1/2}$.
\begin{remark}
If instead of the column sums $\boldsymbol{c}=M^\top \mathbf{1}$ we fix the row sums $\boldsymbol{r}=M\mathbf{1}$, the same procedure applies after the substitutions
\[
\boldsymbol{w} \;=\; D_{\boldsymbol{\pi}}^{-1/2}\mathbf{1}\quad \text{with}\quad \boldsymbol{v} \;:=\; D_{\boldsymbol{\pi}}^{1/2}\mathbf{1},
\qquad
\tilde{\boldsymbol{c}} \;=\; D_{\boldsymbol{\pi}}^{-1/2}\boldsymbol{c}\quad \text{with}\quad \tilde{\boldsymbol{r}} \;:=\; D_{\boldsymbol{\pi}}^{1/2}\boldsymbol{r}.
\]
Indeed, since $m_{ij}=s_{ij}\sqrt{\pi_j}/\sqrt{\pi_i}$, one has
\[
r_i=\sum_{j=1}^k m_{ij}
=\frac{1}{\sqrt{\pi_i}}\sum_j s_{ij}\sqrt{\pi_j}
=\frac{(S\boldsymbol{v})_i}{\sqrt{\pi_i}}
\quad \text{if and only if}\quad
S\boldsymbol{v}=\tilde{\boldsymbol{r}}.
\]
Therefore the feasible set becomes
\[
\mathcal{S}_r \;=\; \{\,S\in\R^{k\times k} : S=S^\top,\ S\ge 0,\ S\boldsymbol{v}=\tilde{\boldsymbol{r}} \,\}.
\]
and the semidefinite program is identical to \eqref{eq:lsp1} with the single change $S\boldsymbol{w}=\tilde{\boldsymbol{c}} \,\ \text{with}\,\ S\boldsymbol{v}=\tilde{\boldsymbol{r}}$. 
The compactness proof also adapts: from $S\boldsymbol{v}=\tilde{\boldsymbol{r}}$ with $v_j>0$,
\[
\sum_{j=1}^k v_j s_{ij}=\tilde r_i
\ \ \Rightarrow\ \
0\le s_{ij}\le \frac{\tilde r_i}{v_j},
\]
so that $\mathcal{S}_r$ is bounded. Hence the problem is solved in complete analogy to the column-sum case.
\end{remark}

\subsection{Details of the computations in Belgium study}
\label{secA2}
We refer to children and adults as age groups 1 and 2, respectively, and the types $CL, CS, AL, AS$ as types $(1,L), (1,S), (2, L), (2,S)$, respectively.  The $4 \times 4$ extended contact matrix is given by
\begin{equation}
\label{equ:C4by4}
C=\begin{bmatrix} C_{11} & C_{12} \\
 C_{21} & C_{22} 
\end{bmatrix}, 
\end{equation}
where, for $i,j \in \{1,2\}$, 
\[
C_{ij}=\begin{bmatrix} c_{ij}^{LL}&c_{ij}^{LS}\\ c_{ij}^{SL}& c_{ij}^{SS}\end{bmatrix},
\]
with, for example, $c_{ij}^{LS}$ being the average number of daily contacts made by a typical type-$(i,L)$ individuals with type-$(j,S)$ individuals.  

The elements of $C_{ij}$ cannot be estimated from social contact study data but the row sums $c_{ij}^L=c_{ij}^{LL}+c_{ij}^{LS}$ and 
$c_{ij}^S=c_{ij}^{SL}+c_{ij}^{SS}$ can be estimated.  (For each contact in a survey, the age of the contacted individual is recorded but not the individual contacted, so whether a contacted individual is type $L$ or type $S$ is unavailable.)
Using the Belgian social contact study, we obtain the estimates
\[
C^L=\begin{bmatrix}  c_{11}^L & c_{12}^L \\
 c_{21}^L & c_{22}^L 
\end{bmatrix}
=
\begin{bmatrix} 3.1715 & 6.2589 \\
 0.6287 & 5.7262 
\end{bmatrix}
\quad\text{and}\quad
C^S=\begin{bmatrix}  c_{11}^S & c_{12}^S \\
 c_{21}^S & c_{22}^S
\end{bmatrix}
=
\begin{bmatrix} 18.1405 & 11.8257 \\
 4.3245 & 25.1809 
\end{bmatrix}.
\]
Since each age group is split into two equally-sized sub-groups, the vector of community fractions 
\[
\boldsymbol{\pi}=(\pi_{CL}, \pi_{CS}, \pi_{AL}, \pi_{AS})^{\top}=(0.1075, 0.1075, 0.3925, 0.3925)^{\top}.
\]
Note that the $2 \times 2$ contact matrix $C$ at~\eqref{equ:C2by2} is given by $(C^L+C^S)/2$.  Further, the row sums of $C$ are given by corresponding row sums of $C^L$ and $C^S$, yielding
\begin{equation}
\label{equ:vecrowC}
\boldsymbol{r}^C=(r_{CL}, r_{CS}, r_{AL}, r_{AS})^{\top}=(9.4304, 29.9662, 6.3549, 29.5125)^{\top}.
\end{equation}

To parameterize the contact matrix $C$, given by~\eqref{equ:C4by4}, it is convenient to let 
\[
A=\begin{bmatrix} A_{11} & A_{12} \\
 A_{21} & A_{22} 
\end{bmatrix}
=D_{\boldsymbol{\pi}}C.
\]
Note that $A$ is symmetric, so we need only parameterize the sub-matrices $A_{11}, A_{12}, A_{22}$, and that the row sums $a_{ij}^L$ and $a_{ij}^S$ (defined analogously to $c_{ij}^L$ and $c_{ij}^S$ above) are known, since $\boldsymbol{\pi}$ is known.  
Exploiting symmetry,
\[
A_{ij}=\begin{bmatrix} \theta_{ij}&a_{ij}^L-\theta_{ij}\\ a_{ji}^L-\theta_{ij}& a_{ij}^S-(a_{ji}^L-\theta_{ij})\end{bmatrix},\qquad(i,j)=(1,1), (1,2), (2,2),
\]
where $\theta_{ij}$ satisfies the constraint
\begin{equation}
\label{equ:thetaijbounds}
\max(0,a_{ji}^L-a_{ij}^S) \le \theta_{ij} \le \min(a_{ij}^L, a_{ji}^L).
\end{equation}
Recall that we assume that the mean infectious period is one.  For a given transmission probability $p$, the next-generation 
matrix $M=pC=pD_{\boldsymbol{\pi}}^{-1}A$.  The `DB: upper (restricted)' and `DB: lower (restricted)' plots in Figure~\ref{fig:Belgian-bound} are obtained by maximizing and minimizing $R_0$ and $\bar{\tau}$ for $M$, with respect to $\theta_{11}, \theta_{12}, \theta_{22}$, under the constraints given by~\eqref{equ:thetaijbounds}, for which we use the 
MATLAB function \emph{fmincon}; see below for further details.

Suppose now that only $\boldsymbol{\pi}$, detailed balance and the row sums $\boldsymbol{r}^C$, given by~\eqref{equ:vecrowC} are assumed.  Then the row sums of $A$ are given by $\boldsymbol{r}^A=D_{\boldsymbol{\pi}}\boldsymbol{r}^C$.  Give the types CL, CS, AL, AS the labels $1,2,3,4$, respectively, and write $\boldsymbol{r}^A=(r_1^A, r_2^A, r_3^A, r_4^A)^{\top}$.  Then
\[
A=\begin{bmatrix}
r_1^A-\theta_1-\theta_2-\theta_3 & \theta_1 & \theta_2 & \theta_3 \\ \medskip
\theta_1 & r_2^A- \theta_1-\theta_4-\theta_5& \theta_4 & \theta_5\\ \medskip
\theta_2 & \theta_4 & r_3^A-\theta_2-\theta_4-\theta_6 & \theta_6\\ \medskip
 \theta_3& \theta_5& \theta_6&r_4^A-\theta_3-\theta_5-\theta_6
    \end{bmatrix},
\]
where $\boldsymbol{\theta}=(\theta_1, \theta_2,\dots, \theta_6)$ satisfies 
\begin{equation}
\label{equ:thetalb}
\theta_i \ge 0 \qquad (i=1,2,\dots,6)
\end{equation}
and
\begin{align}
\theta_1+\theta_2+\theta_3&\le r_1^A \label{equ:thetaub1}\\
\theta_1+\theta_4+\theta_5 &\le r_2^A \label{equ:thetaub2}\\
\theta_2+\theta_4+\theta_6 & \le r_3^A \label{equ:thetaub3}\\
\theta_3+\theta_5+\theta_6 &\le r_4^A \label{equ:thetaub4}
\end{align}
Thus, we now need to maximize and minimize $R_0$ and $\bar{\tau}$ with respect to $\boldsymbol{\theta}$ under the constraints~\eqref{equ:thetalb}-\eqref{equ:thetaub4}, for which we use again the 
MATLAB function \emph{fmincon}.

The function \emph{fmincon} uses an iterative procedure to find a local minimum of the objective function from a given initial value. For the `DB (restricted)' plots, we use random initial values for $\theta_{11}, \theta_{12}, \theta_{22}$, chosen independently and uniformly from the corresponding interval given in~\eqref{equ:thetaijbounds}.   For the `DB' plots we choose a random initial value for $\boldsymbol{\theta}$ as follows.  Let $Y_1, Y_2, Y_3, ,X_1, X_2, \dots, X_6$ be independent and identically distributed unit-mean exponential random variables.  We set $\theta_i=X_i r_i^A/(Y_1+X_1+X_2+X_3)$ $(i=1,2,3)$.
Provided, $a_2=r_2^A-\theta_1 > 0$, we set $\theta_i=X_i a_2/(Y_2+X_4+X_5)$ $(i=4,5)$. Then, provided $a_3= r_3^A-\theta_2-\theta_4 > 0$, we set $\theta_6=a_3 X_6/(Y_3+X_6)$. Finally, we accept the resulting $\boldsymbol{\theta}$ if $r_4^A-\theta_3-\theta_5-\theta_6$.  If at any stage an inequality is not satisfied we start the whole process again from scratch, continuing until we obtain a $\boldsymbol{\theta}$ that satisfies all four inequalities.  (Note that if the types are relabelled so that 
$r_1^A< r_2^A<r_3^A<r_4^A$ then the inequality $a_2> 0$ is necessarily satisfied.)

Since \emph{fmincon} tries to find a local minimum of an objective function, we use several independent random initial values, as described above, for each optimization problem, taking the minimum or maximum as appropriate.  Note that the bounds for $R_0$ are linear in the transmission probability $p$ and zero when $p=0$, so the numerical optimization needs only to be performed for a single value of $p$.  The bounds for $\bar{\tau}$ must be computed separately for each value of $p$.

\end{appendices}

%\end{linenumbers}
\end{document}